%% file: main.tex
\documentclass[lettersize,journal]{IEEEtran}
\usepackage{amsmath,amsfonts,amsthm}
\usepackage{algorithmic}
\usepackage{amssymb}
\usepackage{algorithm}
\usepackage{array}
\usepackage[caption=false,font=scriptsize,labelfont=sf,textfont=sf]{subfig}
\usepackage{textcomp}
\usepackage{stfloats}
\usepackage{url}
\usepackage{verbatim}
\usepackage{graphicx}
\usepackage{cite}
\usepackage{multirow}
\usepackage{pifont}
\usepackage{textcomp}
\usepackage{booktabs}
\usepackage{makecell}
\usepackage{ulem}
\usepackage{threeparttable}
\newtheorem{theorem}{\bf Theorem}
\newtheorem{definition}{\bf Definition}
\newtheorem{corollary}{\bf Corollary}
\newtheorem{assumption}{\bf Assumption}

\makeatletter

\newcommand{\Rmnum}[1]{\expandafter\@slowromancap\romannumeral #1@}
\newcommand{\NOALGCOMMENT}[1]{\noindent/* #1 */}
\newcommand{\NOBLGCOMMENT}[1]{\hspace*{1em} \(\triangleright\) #1}

\makeatother

\input{macros}
\begin{document}
\title{\tool: Local Layer-wise Differential Privacy in Federated Learning}

\author{Yunbo Li,~\IEEEmembership{Student Member,~IEEE,}
        Jiaping Gui,~\IEEEmembership{Member,~IEEE,}
        Fanchao Meng,
        Yue Wu,~\IEEEmembership{Senior Member,~IEEE}

\thanks{Yunbo Li, Jiaping Gui, Fanchao Meng, and Yue Wu are with Shanghai Jiao Tong University, Shanghai, 200240, China (e-mail: \{li-yun-bo, jgui, mactavishmeng, wuyue\}@sjtu.edu.cn).}
}

\markboth{Journal of \LaTeX\ Class Files,~Vol.~14, No.~8, August~2021}%
{Shell \MakeLowercase{\textit{et al.}}: A Sample Article Using IEEEtran.cls for IEEE Journals}


\maketitle

\input{input/abstract}

\input{input/intro}
\input{input/rw}

\input{input/pre}

\input{input/threat}

\input{input/method}

\input{input/conv}
\input{input/exp_ret}

\input{input/discussion}

\input{input/summ}

\normalem
\bibliographystyle{IEEEtran}
\bibliography{IEEEabrv,file}

\input{input/appendix}

\newpage

 




\vfill

\end{document}

%% file: macros.tex
\usepackage{acronym}
\usepackage{xspace}
\usepackage{xcolor}

\newcommand{\ie}{i.e.,\xspace}
\newcommand{\eg}{e.g.,\xspace}

\newcommand{\etal}{\textit{et al.}\xspace}
\newcommand{\eat}[1]{}

\acrodef{iot}[IoT]{Internet of Things}

\newcommand{\change}[1]{\textcolor{cyan}{#1}}

\newcommand{\tool}{{LaDP-FL}\xspace}
\acrodef{sgd}[FedSGD]{Federated Stochastic Gradient Descent}

\acrodef{avg}[FedAvg]{Federated Averaging}

\newcommand{\revise}[1]{\textcolor{black}{#1}}
\newcommand{\mapping}[1]{}

\newcommand{\heading}[1]{\vspace{6pt}\noindent{\underline{\textsc{#1}}}}
\newcommand{\code}{\url{https://anonymous.4open.science/r/LaDP-92D7}}

%% file: input/abstract.tex
\begin{abstract}
Federated Learning (FL) enables collaborative model training without direct data sharing, yet it remains vulnerable to privacy attacks such as model inversion and membership inference. \revise{Existing differential privacy (DP) solutions for FL often inject noise uniformly across the entire model, degrading utility while providing suboptimal privacy-utility tradeoffs. To address this, we propose \tool, a novel layer-wise adaptive noise injection mechanism for FL that optimizes privacy protection while preserving model accuracy. \tool leverages two key insights: (1) neural network layers contribute unevenly to model utility, and (2) layer-wise privacy leakage can be quantified via KL divergence between local and global model distributions. \tool dynamically injects noise into selected layers based on their privacy sensitivity and importance to model performance.}

\revise{We provide a rigorous theoretical analysis, proving that \tool satisfies $(\epsilon, \delta)$-DP guarantees and converges under bounded noise. Extensive experiments on CIFAR-10/100 datasets demonstrate that \tool reduces noise injection by 46.14\% on average compared to state-of-the-art (SOTA) methods while improving accuracy by 102.99\%. Under the same privacy budget, \tool outperforms SOTA solutions like Dynamic Privacy Allocation LDP and AdapLDP by 25.18\% and 6.1\% in accuracy, respectively. Additionally, \tool robustly defends against reconstruction attacks, increasing the FID of the reconstructed private data by $>$12.84\% compared to all baselines. Our work advances the practical deployment of privacy-preserving FL with minimal utility loss.}

\end{abstract}


%% file: input/intro.tex
\section{Introduction}

\IEEEPARstart{F}{ederated} learning (FL)~\cite{mcmahan2017communication} has become a paradigm for collaborative machine learning, which enables participants to exchange model parameters or gradient updates without directly sharing their private data with a central aggregating server. With the capability of preserving participant privacy and addressing the issue of data isolation, FL has attracted substantial interest from researchers in recent years~\cite{karimireddy2020scaffold,li2019convergence,ezzeldin2023fairfed}. \mapping{Comment 1.2 in R1}However, FL systems encounter practical challenges in providing comprehensive privacy protection for all participants to prevent various privacy attacks, such as deep leakage from gradients~\cite{zhao2020idlg, wei2021gradient}, source inference attacks~\cite{hu2023source, li2024subject}, model inversion attacks~\cite{hitaj2017deep, wu2024fedinverse}, and membership inference attacks~\cite{shokri2017membership,hu2022m}. Consequently, the private data of participants remains a significant privacy risk.

To mitigate the aforementioned privacy concerns, researchers have proposed different solutions that can be categorized into two types of mechanisms: homomorphic encryption (HE)~\cite{zhang2020batchcrypt, ma2022privacy} and differential privacy (DP)~\cite{el2022differential, truex2020ldp}. HE is a technique that encrypts the plaintext and ensures that the result of performing operations on the ciphertext is equivalent to that of performing corresponding operations on the plaintext. Through encryption, this mechanism effectively prevents malicious servers from intercepting or stealing model information during transmission. However, HE can incur significant resource overhead~\cite{xue2023differentially}, which limits its application in the real world. In contrast, DP prevents adversaries from inferring private model information by rendering neighboring datasets indistinguishable. Specifically, DP introduces constrained noise to model parameters and then obfuscates them. A typical method to implement DP in FL is to inject noise into local models, also known as local differential private federated learning (LDP-FL)~\cite{truex2020ldp, yuan2023amplitude}. Despite the promise to protect the privacy of each participant, LDP-FL may significantly impair the performance of the global model~\cite{el2022differential} since the noise is accumulated on the server side.

Extensive research~\cite{zhou2022pflf, fu2022adap, wei2023securing, yuan2023amplitude, xue2023differentially} has been conducted to find an optimal trade-off between the privacy protection of local models and the inference performance of the global model in LDP-FL. To achieve this, researchers have designed various adaptive strategies to dynamically adjust the amount of noise injection, leveraging model information such as $\mathcal{L}_2$ norm~\cite{wei2023securing}, clip bound~\cite{fu2022adap}, sensitivity~\cite{xue2023differentially}, and propagation errors~\cite{zhang2021privacy}. This information is determined by comparing the current training round with historical round(s). However, existing approaches estimate privacy guarantees for the entire local model during each iteration, which results in unnecessary noise addition to layers that contain minimal or no private information but are crucial for prediction accuracy. Consequently, the performance of the global model degrades.

To address the aforementioned limitations, we propose \tool, a novel privacy protection approach in FL that leverages layer-wise noise injection in local models. A key insight in \tool is that neural network layers impose varying degrees of impact (\ie importance) on the global model~\cite{jiang2022model, jiang2022fedmp, vogels2019powersgd}. Besides, these layers, providing granular details of local models, can be manipulated independently~\cite{ma2022layer, lee2023layer}. Unlike existing works that solely support noise injection across the entire local model, \tool can selectively identify neural network layers based on their importance and inject varying levels of noise into these layers to protect privacy. 

However, fulfilling the above insights presents two challenges. First, how to inject noise into selective neural network layers to protect privacy while maintaining model accuracy is an open problem. An ideal solution is to inject more noise in layers that contain significant private information but have a relatively small impact on the model's accuracy. However, this is a challenging task pertaining to model interpretability, especially in complex neural network structures. Second, each participant lacks prior knowledge about the optimization direction of the global model. It is difficult to determine the noise injection amount for different layers in a local model.

To tackle the first challenge, \tool leverages a Layer Selection module to determine the importance of each local neural network layer on the global model. Specifically, this module utilizes the weight values of each layer as an indicator of its influence on the inference accuracy of the global model. \mapping{Comment 3.1 in R1}We adopt this strategy because adversaries rely on the predictive capabilities of the global model to launch attacks by training shadow models (e.g., binary classifiers used in membership inference attacks~\cite{shokri2017membership, hu2022membership} or the GANs used in model inversion attacks~\cite{hitaj2017deep,wu2024fedinverse}). If layers within a local model harbor private information yet exert little or no impact on the global model's accuracy, adversaries would face difficulties in training a shadow model that facilitates their attack. To handle the second challenge, \tool employs a Privacy Estimation module that adopts a carefully crafted Kullback-Leibler (KL) divergence-based approach to estimate the amount of private information contained in each layer of the local model. This amount serves as a crucial indicator to quantify the privacy levels of critical layers, guiding the subsequent layer-wise noise injection. We demonstrate that by dynamically injecting noise into selected neural network layers, \tool achieves an optimal balance between privacy protection and model accuracy.

In addition, we provide a thorough theoretical analysis of the privacy and convergence properties of \tool. Specifically, we derive bounds on the parameters that ensure differential privacy guarantees of \tool within Gaussian mechanisms. Under smoothness assumptions, we establish upper bounds on the model’s convergence. Our theoretical analysis reveals that, under specific conditions, \tool approaches the theoretical optimal upper bound as the iteration number increases, thereby providing a rigorous guarantee of its asymptotic optimality.

We conducted a comprehensive evaluation on the efficacy of \tool utilizing the ResNet-18~\cite{he2016deep} and CNN architecture on two prominent datasets: CIFAR-10~\cite{krizhevsky2009learning} and CIFAR-100~\cite{krizhevsky2009learning}. We rigorously evaluated \tool in terms of accuracy, noise injection scale, and resource utilization under various scenarios. Experimental results demonstrate that, \revise{when compared to the traditional Gaussian DP mechanism (Full DP)~\cite{9069945}, Time-Varying DP~\cite{yuan2023amplitude}, and Sensitive DP~\cite{xue2023differentially}, \tool reduces the average noise injection amount by 69.32\%, 51.35\%, and 1.2\%, respectively, while the model accuracy increases by 236.32\%, 144.93\%, and 102.42\%, respectively. Meanwhile, under the same privacy budget, \tool outperforms two state-of-the-art (SOTA) solutions, Dynamic Privacy Allocation LDP~\cite{zhang2024dynamic} and AdapLDP~\cite{10851366}, by an average of 25.18\% and 6.1\%, respectively, in terms of accuracy. Additionally, \tool injects an average of 40.54\% and 68.31\% less noise than Dynamic Privacy Allocation LDP and AdapLDP, respectively. \mapping{Comment 1.4 in R2}By averaging these results, it can be concluded that \tool achieves an average accuracy improvement of 102.99\% across all scenarios, as well as a reduction in average noise injection by 46.14\%. Overall, \tool achieves a superior balance between privacy protection and model accuracy with a feasible time cost. To verify the privacy protection capability of \tool, we further validated its resistance against image reconstruction attacks on FEMNIST~\cite{caldas2018leaf} and CIFAR-10~\cite{krizhevsky2009learning} datasets. The results show that attackers can hardly infer the victim's private data after utilizing \tool's strategy.}

We summarize our contributions as follows:
\begin{itemize}
	\item \revise{\textbf{Layer-wise Privacy Protection Framework.} We propose \tool, the first federated learning framework that provides differential privacy guarantees through adaptive layer-wise noise injection. By selectively perturbing critical neural network layers, \tool achieves finer-grained privacy-utility tradeoffs compared to model-level DP approaches.}

	\item \revise{\textbf{Dynamic Privacy Estimation.} We design a novel KL-divergence-based module to quantify layer-specific privacy risks, enabling adaptive noise scaling proportional to the sensitive information contained in each layer. This addresses the limitations of coarse-grained with static/dynamic noise injection in prior work.}
	
	\item \revise{\textbf{Theoretical Guarantees.} We formally prove that \tool satisfies $(\epsilon, \delta)$-DP under Gaussian mechanisms (Theorem~\ref{thm:mydp}). Furthermore, we establish convergence guarantees under non-IID data distributions (Theorem~\ref{Thm:conv}), demonstrating asymptotic optimality.}
	
	\item \revise{\textbf{Comprehensive Evaluation.} We evaluate \tool on multiple datasets (CIFAR-10/100) and models (CNN, ResNet-18), showing:	(1) \textit{Accuracy:} 102.99\% average improvement over SOTA DP-FL methods. (2) \textit{Efficiency:} 46.14\% reduction in noise injection volume. (3) \textit{Robustness:} Effective defense against privacy reconstruction attacks. (4) \textit{Scalability:} Maintains effectiveness under extreme non-IID settings.}
\end{itemize}

\revise{\textbf{Open Science.} We release the source code of this project at \code.}

%% file: input/rw.tex
\section{Related Works}
\label{sec:2}

Utilizing DP for privacy protection in FL encompasses two primary research directions: Central Differential Privacy Federated Learning (CDP-FL) and Local Differential Privacy Federated Learning (LDP-FL). In this section, we will provide a comprehensive overview of the related work on both directions, including their distinct characteristics and design considerations.

\subsection{Central Differential Privacy Federated Learning.}
\label{sec_rw_cdpfl}
There are several works~\cite{geyer2017differentially, mcmahan2017learning,hu2020personalized} that apply DP techniques by injecting noise on the server side. In such FL settings, the server injects noise into the model after aggregating new local models from all clients and then sends out the updated model to all clients. CDP-FL prevents participants from inferring each other's information by disrupting the model distribution. In addition, CDP-FL poses minimum interference on the global model, thereby preserving its prediction performance. However, CDP-FL does not prevent the server from inferring private information from local models~\cite{xue2023differentially} since the server can directly access client model weights or gradient parameters without any obfuscation. In contrast, our work primarily employs a local noise injection technique, which provides stronger privacy protection capabilities and ensures that local models cannot be analyzed by the server.

\subsection{Local Differential Privacy Federated Learning.}
\label{sec_rw_ldpfl}
\mapping{Comment 1.3 in R1}Compared to CDP-FL, LDP-FL offers a safer protection mechanism and has garnered more extensive research attention~\cite{kim2021federated, yang2023local, cai2024plfa}, spanning fields such as wireless communication~\cite{seif2020wireless} and medical analysis~\cite{xie2024privacy}. Several studies~\cite{wang2023ppefl,hu2023federated} have adopted a ``first-compression-then-perturbation'' approach, which reduces communication overhead while decreasing the dimensionality of noise injection. 
Chen \etal~\cite{chen2024clfldp} further consider privacy allocation among different parties within compression. 
Yang \etal~\cite{yang2023dynamic} focus on the personalized federated learning framework and inject noise only on shared parameters by utilizing Fisher information.
However, all these approaches inject noise with a fixed distribution into their models. Over multiple iterations, excessive noise may be introduced, resulting in an increased signal-to-noise ratio within the model and subsequently impeding its performance.

\mapping{Comment 1.3 in R1}To address this issue, researchers have explored methods for dynamically controlling noise injection quantities in each round to achieve a better privacy-utility tradeoff~\cite{wei2023securing,yuan2023amplitude,zhang2022understanding, zhang2024dynamic}. Some works~\cite{xue2023differentially,han2021accurate} adjust noise scales by gradient clipping boundaries or historical model information. Fu \etal~\cite{fu2022adap} propose Adap DP-FL, which injects noise into local gradients with adaptive scale adjustment, through clipping-bound updates based on previous bounds and model evolution. Lin \etal~\cite{lin2023heterogeneous} introduce HDP-FL, which leverages contracts to incentivize each client to obtain the optimal amount of private data and more accurately inject dynamic noise. All these works adopt a coarse-grained approach, treating local neural network models as a single entity and injecting noise at the model level. However, this can lead to substantial performance degradation of the global model. In contrast, we propose injecting noise with finer granularity, which better balances model accuracy and privacy protection.

%% file: input/pre.tex
\section{Background} 
\label{sec:3}
\subsection{Federated Learning}

FL systems typically consist of an aggregation server and multiple clients. 
Let $S$ denote the aggregation server and $\mathcal{C} = \{\mathcal{C}_1, \mathcal{C}_2,...,\mathcal{C}_N\}$ denote $N$ clients, each of which has a corresponding dataset $\mathcal{D}_i$, where $i \in \{1, 2, ..., N\}$. Each data point $k \in \mathcal{D}_i$ is denoted by $(x_k, y_k)$. The server has a global model with weights denoted as $w_g$ and a collection $\mathcal{S}$ that stores all the data the clients upload.
Furthermore, we use $w_i^t$ to denote the local model parameters and $F(w_i^t)$ to represent the empirical risk function of the model in $C_i$ at the $t^{th}$ global round. Therefore, $C_i$ can use a prediction loss function $l$ to train its local model as: 
\begin{equation}\label{eq:loss-function}
    \min_{w_i} \ F(w_i^t) \triangleq \frac{1}{|\mathcal{D}_i|} \sum_{k \in \mathcal{D}_i} l(k;w_i^t) = \frac{1}{|\mathcal{D}_i|} \sum_{k \in \mathcal{D}_i} l(y_k, w_i^t(x_k)).
\end{equation}

The above loss function is typically optimized using the local stochastic gradient descent (SGD) algorithm as:
\begin{equation}\label{eq:sgd-function}
    w_{i, e} = w_{i, e-1} - \eta \nabla F(w_{i, e-1}),
\end{equation}
where $e$ denotes the local training epoch and $\eta$ is the learning rate. Note that we follow the same format as SCAFFOLD~\cite{karimireddy2020scaffold} to represent the local epoch (\eg $e$ in Equation~\ref{eq:sgd-function}) and global round (\eg $t$ in Equation~\ref{eq:loss-function}). 

In synchronous federated learning (SFL), the server employs a randomized selection process to determine a set of active clients, denoted by $\mathcal{AC} \subseteq \mathcal{C}$. Subsequently, all the activated clients upload their local models after completing local training. While the server receives shared data from all clients in $\mathcal{AC}$, it will update the global model~\cite{mcmahan2017communication} by:
\begin{equation}
    w_g^{t} = \frac{1}{D} \sum_{w_i \in \mathcal{S}} |\mathcal{D}_i| w_i^t, \quad D \triangleq \sum_{w_i^t \in \mathcal{S}} |\mathcal{D}_i|.
\label{equ:fedavg_update}
\end{equation}

\subsection{Differential Privacy}
DP quantifies the degree of divergence between two neighboring data points within a dataset, indicating the adversary's capability to discern the data points based on the disclosed information. Based on this concept, we can establish a mathematical benchmark for assessing the privacy of a dataset.

\begin{definition}[$(\epsilon, \delta)$-DP~\cite{abadi2016deep}]
     Let $f: \mathcal{D} \longrightarrow \mathcal{R}$ be an arbitrary function. If for all of the adjacent datasets\footnote{Two datasets $\mathcal{D}$ and $\mathcal{D}'$ are considered adjacent if they differ by only one sample. Conventionally, adjacent datasets $\mathcal{D}'$ can be derived by either deleting or replicating a single instance from the original dataset $\mathcal{D}$. For a dataset $\mathcal{D}$ consisting of $n$ instances, there exist precisely $2n$ distinct adjacent datasets.} $\mathcal{D}$, $\mathcal{D}'$ and any randomized output $s \subseteq \mathcal{R}$, the function $f$ satisfies: 

    \begin{equation}
        Pr[f(\mathcal{D}) \in s] \leq e^{\epsilon}Pr[f(\mathcal{D}') \in s] + \delta,
        \label{equ:epsilon_delta_dp}
    \end{equation}
    where $\delta$ is the failure probability. 

    Then we say that $f$ is $(\epsilon, \delta)$-differentially private. 
    \label{def:ed_dp}
\end{definition}

$(\epsilon, \delta)$-DP ensures that for all adjacent $\mathcal{D}, \mathcal{D}'$, the privacy loss will be bounded by $\epsilon$ with probability at least $1-\delta$. If $\delta = 0$, we say that $f$ is $\epsilon$-differentially private. 

The essence of the noise injection mechanism lies in incorporating noise sourced from a predefined distribution, thereby aligning the data distribution with DP specifications. This strategy aids in obscuring discrepancies between neighboring datasets and fortifying the processed dataset against potential privacy breaches. \mapping{Comment 3.4 in R1}A commonly adopted method involves injecting Gaussian noise, which provides a more relaxed form of DP protection, as proven by Theorem~\ref{thm:2}~\cite{dwork2014algorithmic}.

\begin{theorem}
\label{thm:2}
    Let $\epsilon \in (0, 1)$ and $\delta$ be arbitrary. For $c^2 \geq 2ln(\frac{1.25}{\delta})$, the Gaussian Mechanism $n \sim \mathcal{N}(0,\sigma^2)$ with parameter $\sigma \geq \frac{c \Delta f}{\epsilon}$ is $(\epsilon, \delta)-$differentially private~\cite{dwork2014algorithmic}. The sensitivity $\Delta f$ is defined as
\begin{equation}
    \Delta f = \mathop{\max}\limits_{x,x'}||f(x) - f(x')||,
\end{equation}
where $x, x'$ are the arbitrary adjacent inputs.
\end{theorem}

%% file: input/threat.tex
\section{Threat Model}
\label{sec:threat}

We consider a general FL system consisting of a central aggregation server and a collection of distributed clients, as shown in Figure~\ref{fig_1}. Below, we discuss the security assumptions and threats for both the server and the individual clients.

\begin{figure}[!tp]
\centering
\includegraphics[width=3.5in]{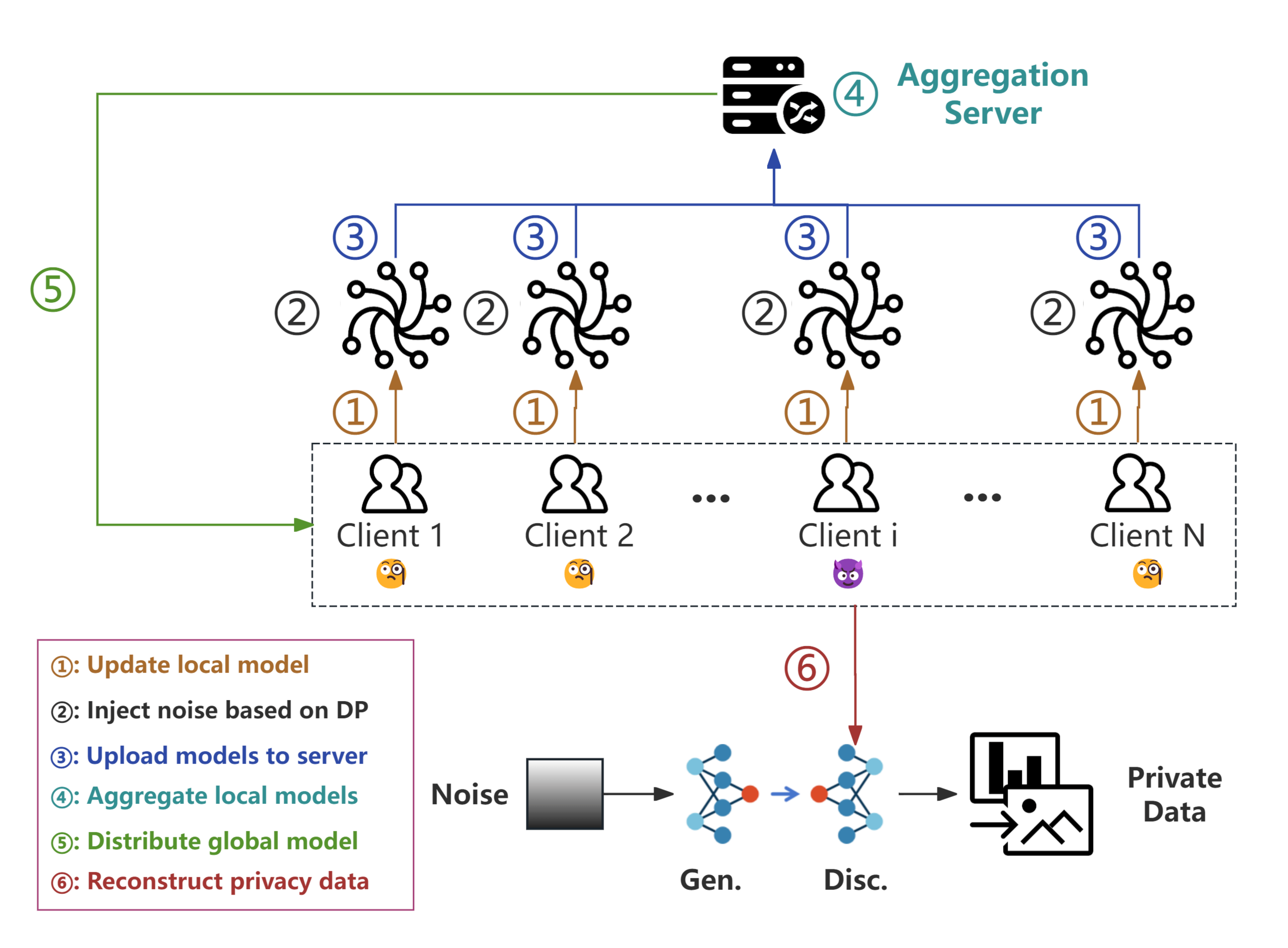}
\caption{An FL model with an honest-but-curious client trying to infer private information from the global model. ``Gen.'' and ``Disc.'' denote Generator and Discriminator, respectively.}
\label{fig_1}
\end{figure}

\begin{itemize}
    \item \textbf{Aggregation Server.} We assume that the central aggregation server is fully trusted, adhering to its designated role in the FL system. Specifically, we assume that the server faithfully performs the tasks of aggregating model parameters (Step \ding{195} in Figure~\ref{fig_1}) and distributing correct and complete models without attempting to compromise user privacy or maliciously altering model parameters (Step \ding{196} in Figure~\ref{fig_1}). Addressing malicious server scenarios is out of the scope of this paper. 
    \item \textbf{Non-malicious Clients.} We consider most of the clients participating in the FL system to be non-malicious, such as \texttt{Client 1} and \texttt{Client N} in Figure~\ref{fig_1}. These clients use local datasets to train models (Step \ding{192} in Figure~\ref{fig_1}) and then upload model parameters to the server (Step \ding{194} in Figure~\ref{fig_1}), awaiting subsequent global updates.
    \item \textbf{Honest-but-Curious (HBC) Clients.} We assume that there exist HBC clients (\eg \texttt{Client i} in Figure~\ref{fig_1}) in the FL system who faithfully execute their designated tasks and upload accurate model parameters. However, these clients exhibit curiosity and seek to reconstruct the private data of other participants. In this paper, we follow related work~\cite{hitaj2017deep} and assume that HBC clients employ a Generative Adversarial Network (GAN) to reconstruct the sensitive information of fellow clients (Step \ding{197} in Figure~\ref{fig_1}). The discriminator architecture of this GAN aligns entirely with the global model, while the generator architecture adopts a deconvolutional neural network. 
\end{itemize}

%% file: input/method.tex
\section{Methodology}
\label{sec:4}
In this section, we first present an overview of \tool, followed by a more detailed description of its design modules. 

\begin{algorithm}[!t]
\caption{\mapping{Comment 3.2 in R1}The workflow of \tool on the server side.}
\footnotesize
    \begin{algorithmic}[1]
        \REQUIRE Global model $w_{g}^t$; Global training round $T$; Server storage collection $\mathcal{S}$.
	\ENSURE A protected model $\tilde{w}_{i}^{t}$.
            \STATE Initialize the global model $w_g^0$ and global round $t \leftarrow 0$;
            \STATE Distribute the global model $w_g^0$;
            \WHILE{Global round $t \leq T$} 
                \STATE Rondamly activate a subset of $\mathcal{C}$, denote as $\mathcal{AC}$;
                \STATE Wait for all the clients in $\mathcal{AC}$ upload their protected model $w_i^t$ by Algorithm 2;
                \STATE  \NOALGCOMMENT{\textit{Aggregate the global model}} \newline \hspace*{1em} $w_g^t = \sum_{w_i^t \in \mathcal{S}} \frac{|\mathcal{D}_i|}{D} \tilde{w}_i^t$;  
                \STATE Distribute new model $w_{g}^{t}$ to all clients.
                \STATE $t \leftarrow t+1$;   
            \ENDWHILE
    \end{algorithmic}
\label{alg:alg1}
\end{algorithm}

\begin{algorithm}[!t]
\caption{\mapping{Comment 3.2 in R1}The workflow of \tool on the client side at global training round $t$.}
\footnotesize
    \begin{algorithmic}[1]
        \REQUIRE New global model $w_{g}^{t-1}$; Local dataset $\mathcal{D}_i$; Local training epoch $E$; Learning rate $\eta$; Privacy parameters $\epsilon, \delta, c_i$; Privacy estimation bound $B$; Gradients clip bound $G_c$; Layer selection parameters $R$.
	\ENSURE A protected model $\tilde{w}_{i}^{t}$.
            \STATE  \NOALGCOMMENT{\textit{Update local model by new global model}} \newline \hspace*{1em}
            $w_{i,0}^t \leftarrow w_g^{t-1}$;
            \STATE  Initialize $e = 0$;
            \WHILE{local epoch $e \leq E$} 
                \STATE Compute the loss $\nabla F(w_{i,e-1}^t)$ by local dataset $\mathcal{D}_i$;
                \STATE  \NOALGCOMMENT{\textit{Clip the local gradients}} \newline \hspace*{1em} $\nabla F(w_{i,e-1}^t) \leftarrow \frac{\nabla F(w_{i,e-1}^t; \mathcal{D}_i)}{ \max(1, \frac{||\nabla F(w_{i,e-1}^t; \mathcal{D}_i)||}{G_c})}$;  
                \STATE \NOALGCOMMENT{\textit{Update local models}} \newline \hspace*{1em}$w_{i,e}^t \leftarrow w_{i,e-1}^t - \eta \nabla F(w_{i,e-1}^t)$;  
                \STATE $e \leftarrow e+1$;   
            \ENDWHILE
            \STATE  \NOALGCOMMENT{\textit{Finish local training}} \newline \hspace*{1em}
            $w_{i}^t \leftarrow w_{i,E}^{t}$;
            \STATE \NOALGCOMMENT{\textit{Estimate the function sensitivity}} \newline  \hspace*{1em}$\Delta f_i = 2\eta E G_c$; 
            \FORALL{layers $j$ in the network $w_i^t$}
                    \IF{Layer selection: $||w_{i,j}^t|| \geq R$} 
                        \STATE \NOALGCOMMENT{\textit{Privacy Estimation by Algorithm~\ref{alg:PE}}} \newline \hspace*{1em}$P_{i,j} \leftarrow \text{\textit{PRIVACY\_ESTIMATION}}(w_{i,j}^t, w_{g}^t)$; 
                        \STATE \NOALGCOMMENT{\textit{Adaptive Noise Injection by Algorithm~\ref{alg:NI}}} \newline \hspace*{1em}$\tilde{w}_{i,j}^{t} \leftarrow \text{\textit{NOISE\_INJECTION}}(w_{i,j}^t,\epsilon,\Delta f_{i,j},c_i,P_{i,j})$; 
                    \ENDIF
            \ENDFOR
    \STATE Upload local model $\tilde{w}_{i}^{t}$.
    \end{algorithmic}
\label{alg:alg2}
\end{algorithm}

\subsection{System Overview}

\mapping{Comment 3.2 in R1}The workflow of \tool on both the server side and the client side is illustrated in Algorithms~\ref{alg:alg1} and~\ref{alg:alg2}, respectively. Specifically, in Algorithm~\ref{alg:alg1}, the server first initializes the federated settings and distributes the initial model to all clients (Lines 1-2). Then, during each global training round, the server randomly selects a subset of activated clients to train the model using their local datasets and awaits their updates (Lines 4-5). Finally, the server aggregates the updates to form a new global model, distributes this new model to all clients, and proceeds to the next global round (Lines 6-8). In Algorithm~\ref{alg:alg2}, upon receiving the initial model from the server (Line 1), each activated local client first executes local training (Lines 2-9). Then, the client estimates the sensitivity of the current epoch (Line 10), and iterates through each neural network layer to determine whether it contains essential information for model updates (Lines 11-12). If the norm value of the weights in a certain neural network layer exceeds a predefined threshold $R$, the client estimates the privacy level $P_{i,j}$ for that layer, and then enforces an upper bound on $P_{i,j}$ to guarantee fundamental privacy protection (Line 13). Next, based on the derived privacy levels, the client determines the noise range and samples noise to inject into each layer (Line 14). Finally, after processing all neural network layers, the client uploads the noisy model to the server (Line 19).

\mapping{Comment 2.3 in R1 \& Comment 3.5 in R1}To better balance privacy protection and model accuracy, \tool ignores those layers with smaller weights when injecting noise (Line 12 in Algorithm~\ref{alg:alg2}). Since pruning away small-weighted layers will not significantly impair the model's predictive ability~\cite{jiang2022model, vogels2019powersgd}, these smaller weights indicate weaker predictive power. Therefore, \tool can mitigate the adverse effects of noise on the model's utility while ensuring that privacy is not unduly compromised. For other layers, \tool first assesses their private content using the KL divergence (Lines 13-14 in Algorithm~\ref{alg:alg2}). For layers containing more private information, \tool will prioritize privacy protection by sacrificing some model utility; otherwise, \tool will minimize noise injection to enhance model utility (Lines 15-16 in Algorithm~\ref{alg:alg2}).

\subsection{System Design}
\mapping{Comment 3.3 in R1}\tool consists of three modules: layer selection, privacy estimation, and noise injection, as illustrated in Figure~\ref{fig_system_model}. We explain each of these modules in more detail below.

\begin{figure}[!tp]
\centering
\includegraphics[width=3.5in]{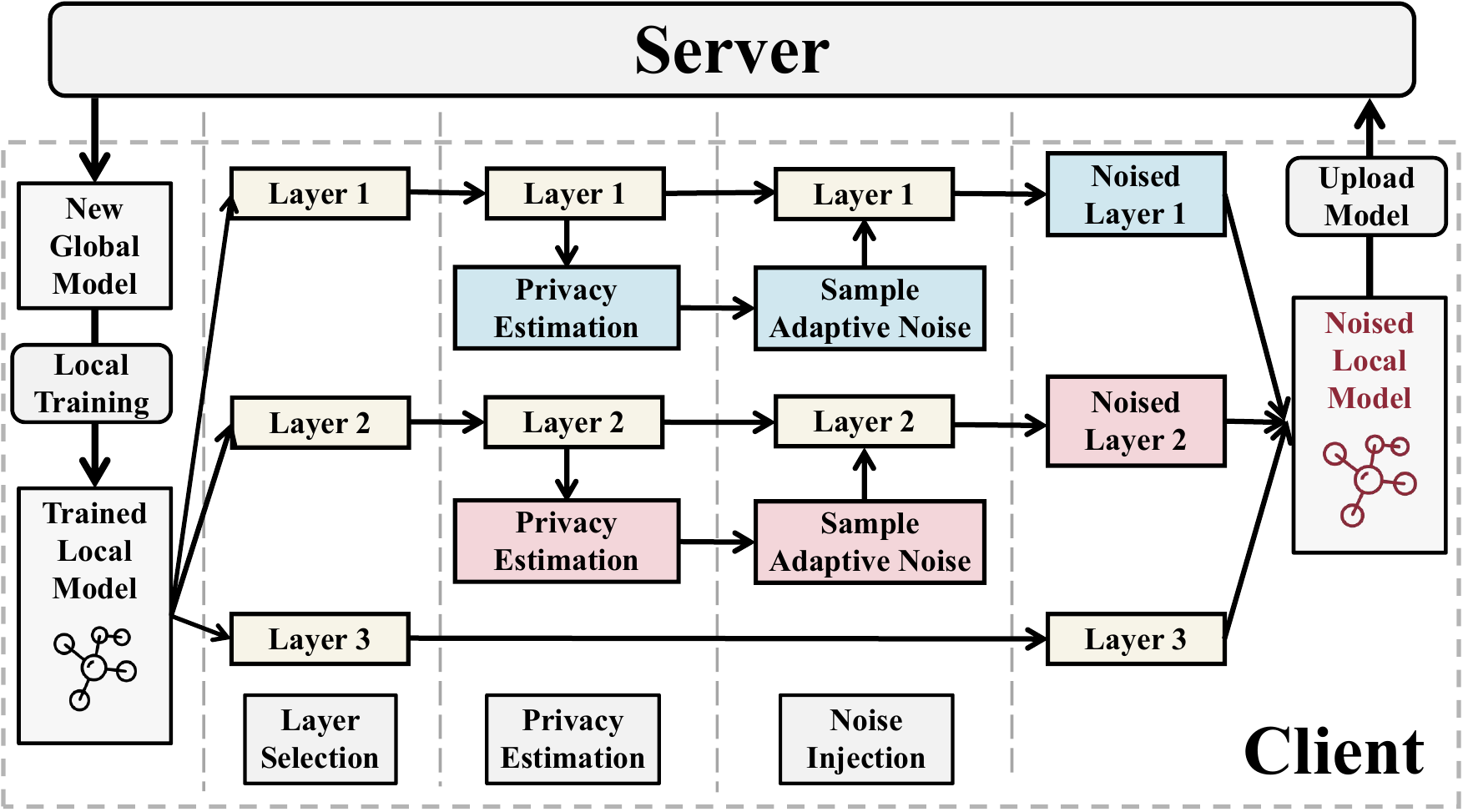}
\caption{\mapping{Comment 3.3 in R1}A visual representation of each module of \tool. The blue and red squares represent different levels of privacy information associated with different layers, resulting in different amounts of noise injection.}
\label{fig_system_model}
\end{figure}

\subsubsection{Layer Selection}
\label{module1}
In this module, \tool identifies the neural network layers that have a more substantial impact on the predictive performance of the global model. Specifically, upon receiving the latest global model from the server, the client $C_i$ first updates its current local model $w_{i,0}^t$ by replacing it with the previous round's global model $w_{g}^{t-1}$.  Then, the client employs its local optimizer to update the model for $E$ epochs iteratively and gets a new model $w_{i, E}^t$. Inspired by related work~\cite{zhang2022understanding}, we clip the model's gradients by $G_c$ based on Assumption~\ref{Ass:2}, which ensures the convergence stability of the model \eat{on heterogeneous data distributions }during local training.

\begin{assumption}[Client-level Bounded Gradients~\cite{zhang2012communication,stich2018local}]  
	\label{Ass:2}
	For $\forall \  C_i \in \mathcal{C}$, we assume that there exists a constant $G_c \geq 0$ such that its gradient $\nabla F(w_i^t)$ satisfies: 
	\begin{equation}
		||\nabla F(w_i^t)|| \leq G_c.
	\end{equation}
\end{assumption}

After completing local training, participants inspect each layer of the neural network to identify those containing significant information. Researchers have employed ablative experiments~\cite{zeiler2014visualizing, yosinski2014transferable} by removing or replacing specific layers to observe changes in model performance. However, utilizing this strategy to select layers is practically infeasible. This is because conducting control experiments for every neural network layer during the FL process could incur enormous time and resource costs.  To address this challenge, \tool leverages a key insight commonly used for model pruning and compression: \textit{layers with larger weights contribute more significantly to the model's prediction}~\cite{jiang2022model, vogels2019powersgd}. Therefore, by calculating the $\mathcal{L}_2$ norm of each layer and considering those exceeding a threshold $R$ as crucial for maintaining model predictability, \tool mitigates the impact of noise on the global model. Besides having a significant impact on the model's prediction, these crucial layers also carry important private information. Intuitively, for layers with smaller weights, it is difficult for adversaries to extract useful information from these layers to launch attacks, such as model inversion~\cite{fredrikson2015model, hitaj2017deep} and membership inference~\cite{shokri2017membership, nasr2019comprehensive, liang2023egia}. Hence, \tool relies on crucial layers for adaptive noise injection, which ensures privacy protection.

\subsubsection{\mapping{Comment 1.5 in R1}Privacy Estimation}
\label{module2}
Traditional privacy attacks~\cite{fredrikson2015model, hitaj2017deep, shokri2017membership} rely on an aggregated global model to infer private information. \revise{Existing DP-FL approaches~\cite{han2021accurate,fu2022adap,yuan2023amplitude} typically inject uniform noise across all model layers, ignoring the varying risks of different layers to privacy leakage. This one-size-fits-all strategy leads to either excessive utility loss (when over-protecting insensitive layers) or insufficient privacy guarantees (when under-protecting critical layers).}

\revise{To address this problem, we employ KL divergence for layer-wise privacy estimation based on three fundamental principles:}

\begin{enumerate}
	\item \revise{\textbf{DP as Distribution Alignment}: Differential privacy fundamentally requires neighboring datasets to produce statistically similar outputs~\cite{dwork2014algorithmic}. The KL divergence directly measures this similarity through the relative entropy between distributions:
	\begin{equation}
		KL(p||q) = \mathbb{E}_p\left[\log\frac{p(x)}{q(x)}\right]
		\label{equ:KL}
	\end{equation}}

	\item \revise{\textbf{Information-Theoretic Interpretation}: The KL value $KL(w_{i,j}^t||w_{g,j}^t)$ quantifies the minimum additional information (in nats) needed to encode layer $j$'s weights using the global distribution rather than the local one. This aligns with privacy leakage metrics in~\cite{sreekumar2023limit}.}

	\item \revise{\textbf{Attack Surface Correlation}: As demonstrated in~\cite{zang2024detection}, reconstruction attack success probability grows exponentially with decreasing KL divergence between victim and attacker models.}
\end{enumerate}

\revise{Below, we consider two scenarios for intuitive interpretation:
\begin{itemize}
	\item \textit{Small KL divergence}: Layer weights closely match the global distribution → Higher privacy risk (requires more noise)
	\item \textit{Large KL divergence}: Significant deviation from global distribution → Lower privacy risk (needs less protection)
\end{itemize}}

\begin{figure}[t]
	\centering
	\includegraphics[width=1\linewidth]{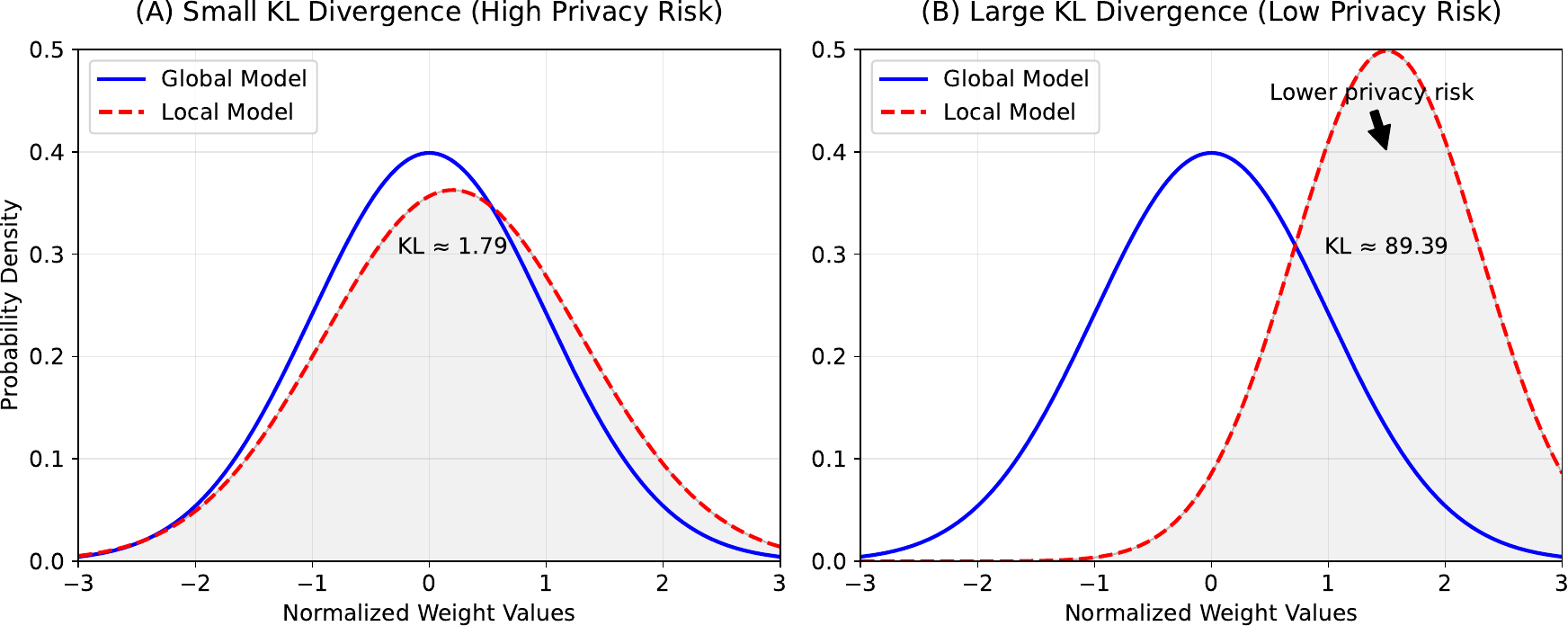}
	\caption{\revise{Illustration of KL divergence's relationship with privacy risk. (A) When $KL \approx 0$, the global model has a higher privacy risk because it closely resembles the local model. (B) When $KL \gg 0$, the local layer reveals little about the global mode due to the significant deviation.}}
	\label{fig:kl_interpretation}
\end{figure}

\begin{algorithm}[!t]
\caption{\mapping{Comment 1.1 in R2}The process of privacy estimation.}
\footnotesize
    \begin{algorithmic}[1]
        \REQUIRE The parameter matrix of the $j$-th layer within the local model $w_i^t$, denoted as $w_{i,j}^t$; Global model $w_{g}^t$.
	\ENSURE The quantified privacy of the $j$-th layer $P_{i,j}$.
            \STATE Extract the parameter matrix of the $j$-th layer within the global model, denoted as $w_{g,j}^t$;
            \STATE Flatten the parameter matrices $w_{i,j}^t$ and $w_{g,j}^t$;
            \newline \NOBLGCOMMENT{\ie flatten to 1-D vectors, such as from $\mathbb{R}^{m \times n}$ to $\mathbb{R}^{mn}$.}
            \STATE Normalize the layers using \textit{softmax()};
            \newline \NOBLGCOMMENT{\eg $[1,2,3] \to [0.09,0.24,0.67]$}
            \STATE Calculate $P_{i,j}$ based on the KL divergence using Equation~\ref{equ:privacy-estimation};\newline \NOBLGCOMMENT{\eg using function \textit{kl\_div()} in the package \textit{torch.nn.functional}.}
            \STATE Clip $P_{i,j}$ by the boundary $B$. 
            \newline \hspace*{1em} $P_{i,j} \leftarrow \min(P_{i,j}, B)$.        
    \end{algorithmic}
\label{alg:PE}
\end{algorithm}

\mapping{Comment 1.1 in R2}Algorithm~\ref{alg:PE} shows the detailed procedure for privacy calculation. \tool first extracts the parameter matrix $w_{g,j}^t$ of the $j$-th layer within the global model (Line 1). Then, \tool flattens and normalizes the selected layers (Lines 2 and 3) to ensure the non-negativity of the KL divergence. Finally, \tool computes $P_{i,j}$ using Equation~\ref{equ:privacy-estimation} (Line 4) and clips it to an upper bound $B$. (Line 5) A large $P_{i,j}$ suggests that the weight distribution in the $j$-th layer of the local model exhibits a substantial deviation compared to its counterpart in the global model, thereby better safeguarding privacy. This is because the substantial deviation makes it difficult for HBC participants to recover information about the $j$-th layer of the victim model from the global model. Hence, \tool leverages $P_{i,j}$ to adjust the noise injection amount adaptively. In particular, as $P_{i,j}$ increases, \tool reduces the noise injection amount (\ie $\sigma_{i,j} \propto 1/P_{i,j}$) to the corresponding layers to protect privacy. \tool also controls the boundary of $P_{i,j}$ to prevent situations where the noise injection amount is too small and tends to zero (see Section~\ref{module3}).

\begin{definition}[\revise{Layer Privacy Metric}]
	For layer $j$ in client $i$'s model $w_i^t$, the privacy content $P_{i,j}$ is:
	\begin{equation}
		P_{i,j} = \min(KL(w_{i,j}^t||w_{g,j}^t), B),
		\label{equ:privacy-estimation}
	\end{equation}
	where $B$ is a clipping boundary ensuring minimum noise injection.
\end{definition}

\revise{This approach provides three key advantages: (1) \textbf{Data-Agnostic}: Requires no knowledge of other clients' data distributions. (2) \textbf{Layer-Adaptive}: Captures varying sensitivity across network depths. (3) \textbf{Attack-Resistant}: Directly correlates with reconstruction difficulty (verified in Section~\ref{sec:defense-eval}).}

\subsubsection{\mapping{Comment 1.7 in R1 \& Comment 2.3 in R1 \& Comment 1.5.1 in R2}Noise Injection}
\label{module3}

\revise{The noise injection process critically impacts both privacy protection and model utility in FL systems. Excessive noise degrades model accuracy and training stability, while insufficient noise compromises privacy. To address these challenges, \tool implements an adaptive layer-wise noise injection mechanism with three key innovations: (1) privacy-aware noise scaling, (2) dimensional consistency enforcement, and (3) cross-client compatibility.}

\revise{First, we conduct a layer-wise sensitivity analysis and derive layer-specific sensitivity bounds, taking into account the non-IID data distribution across clients, based on Theorem~\ref{theorem_sensitivity}.}

\begin{theorem}
	For client $\mathcal{C}_i$ with local gradient $\nabla F(w_i^t)$ clipped by $G_c$, the sensitivity of layer $w_{i,j}$ satisfies:
	\begin{equation}
		\Delta f_{i,j} \leq 2\eta E G_c.
	\end{equation}
	where $\eta$ is the learning rate and $E$ is local epochs.
	\label{theorem_sensitivity}
\end{theorem}

\revise{This elucidates how to estimate the sensitivity in Theorem~\ref{thm:2}, thereby determining the noise variance by the Gaussian mechanism (proof sketch in Supplemental Material).}

\begin{algorithm}[!t]
	\caption{\mapping{Comment 1.1 \& 1.5.1 in R2}Layer-wise Adaptive Noise Injection}
	\footnotesize
	\begin{algorithmic}[1]
		\REQUIRE Layer parameters $w_{i,j}^t \in \mathbb{R}^{m_1\times...\times m_d}$, $\epsilon$, $\Delta f_{i,j}$, $c_i$, $P_{i,j}$.
		\ENSURE Noisy layer $\tilde{w}_{i,j}^{t}$.
		\STATE Compute $\sigma_{i,j} = \frac{c_i \Delta f_{i,j}}{\epsilon P_{i,j}}$
		\STATE Generate noise matrix $n_{i,j}$ where:
		\begin{equation}
			n_{i,j}^{(k_1,...,k_d)} \sim \mathcal{N}(0, \sigma_{i,j}^2)\ \forall k_l \in {1,...,m_l} \nonumber
		\end{equation}
		\STATE Element-wise addition: $\tilde{w}_{i,j}^{t} = w_{i,j}^t + n_{i,j}$
	\end{algorithmic}
	\label{alg:NI}
\end{algorithm}

\mapping{Comment 1.5 in R2}\revise{Algorithm~\ref{alg:NI} details our dimensional-consistent noise injection process. It includes two key implementation aspects: (1) \textbf{Dimensional Matching}: For different selected layers (\eg convolutional layer with $64\times64\times3\times3$ dimensions in ResNet-18), we generate noise tensors matching the parameter dimensions exactly. (2) \textbf{Cross-client Consistency}: All clients within the same FL select layers from the same model architecture, thus ensuring that all distributed clients generate identically shaped additive noise tensors for the same selected layers.}

\revise{\tool adopts an adaptive noise scaling strategy, where the standard deviation of the noise follows:
\begin{equation}
	\sigma_{i,j} = \frac{c_i \Delta f_{i,j}}{\epsilon P_{i,j}},
	\label{equ:adaptive_dp}
\end{equation}
where $P_{i,j}$ is the clipped KL-divergence privacy estimate. This creates an inverse relationship where:
\begin{itemize}
	\item Low $P_{i,j}$ (similar layers) $\Rightarrow$ Larger noise (stronger privacy)
	\item High $P_{i,j}$ (dissimilar layers) $\Rightarrow$ Smaller noise (preserves useful information)
\end{itemize}}

\mapping{Comment 1.7 in R1 \& Comment 1.6 in R2}\revise{\tool ensure robust performance under non-IID data distribution because of two mechanisms: (1) \textbf{Privacy Thresholding}: A minimum noise level $B$ prevents under-protection. (2) \textbf{Dynamic Privacy Estimation}: Models exhibit variations in layer weights across different data distributions; dynamic privacy estimation mechanisms capture these variations and adaptively regulate noise injection to maintain data privacy and model utility concurrently. As demonstrated in Section~\ref{sec:impact_of_noniid}, \tool achieves an average of 25.21\% accuracy improvement over all baselines in extreme non-IID settings.}

%% file: input/conv.tex
\section{\mapping{Comment 1.4 in R1}Theoretical Analysis}
\label{sec:5}

In this section, we conduct a theoretical analysis of the privacy guarantees offered by \tool. We further conduct the convergence analysis of \tool in scenarios involving restricted gradients and noises, demonstrating that \tool has the capability to converge toward the optimal solution.

\subsection{Privacy Analysis}

Based on Theorem~\ref{thm:2}, we have established that under standard Gaussian mechanisms, $(\epsilon, \delta)$-DP is satisfied when the variance satisfies $\sigma \geq \frac{c\Delta f}{\epsilon}$, where $c^2 > 2\ln{\frac{1.25}{\delta}}$. To provide the theoretical privacy guarantees, we first present the following theorem:
\begin{theorem} Given privacy budget $\epsilon$, failure probability $\delta$ and the KL-privacy estimation bound $B$, if the parameter $c_i$ for the client $\mathcal{C}_i$ satisfies:
\begin{equation}
     c_i \geq 
         \begin{cases}
             & \displaystyle{\frac{(\sqrt{\ln{\frac{2}{\pi\delta^2}}} + \sqrt{\ln{\frac{2}{\pi\delta^2}} + 8\epsilon})B}{4}}, \quad \text{if $0 < \delta \leq \sqrt{\frac{2}{\pi e^4}}$} \\
             & \\
             & \displaystyle{\frac{(1 + \sqrt{1 + 2\epsilon})B}{2}}, \quad \text{if $\sqrt{\frac{2}{\pi e^4}} < \delta < 1$.}
         \end{cases}
         \label{thm:equ:my_dp}
\end{equation}
Then, our proposed system \tool satisfies $(\epsilon, \delta)$-DP if the client $\mathcal{C}_i$ adds Gaussian noise with variance in Equation~\ref{equ:adaptive_dp}.
\label{thm:mydp}
\end{theorem}

\begin{proof}
    See Section~\ref{sec:appendixB} in the Supplemental Material.
\end{proof}

\mapping{Comment 3.4 in R1}Theorem~\ref{thm:mydp} ensures that each round of noise injection satisfies $(\epsilon, \delta)$-DP guarantees. In the entire FL training process, each client requiring privacy protection injects noise into their updates according to Theorem~\ref{thm:mydp}'s constraints and then uploads these updates to the server. Therefore, the entire training process can theoretically guarantee DP protection.

\subsection{Convergence Analysis}

First, we introduce some assumptions, standard in the federated learning literature, for our convergence analysis.
\begin{assumption}[L-Smoothness~\cite{li2019convergence}] 
\label{Ass:1}
Objective function $F(x)$ satisfies $L$-Lipschitz Smoothness, which means that for $\forall x, y$, 
     \begin{equation}
         F(y) - F(x) \leq   \nabla F(x)^T(y-x) + \frac{L}{2}||y-x||^2,
     \end{equation}
     where $L > 0$ is a constant.
\end{assumption}

\begin{assumption}[$\mu$-Polyak-Lojasiewicz condition~\cite{karimi2016linear}] 
\label{Ass:3}
Objective function $F(x)$ satisfies $\mu$-Polyak-Lojasiewicz condition, \ie for $\forall x$, 
     \begin{equation}
         F(x) - F^* \leq   \frac{1}{2\mu}||\nabla F(x)||,
     \end{equation}
     where $\mu > L$ is a constant and $F^*$ is the global optima of $F(x)$.
\end{assumption}

\begin{assumption}[Bounded Noise] 
\label{Ass:4}
To guarantee the effectiveness of the training, we assume that the amount of noise $n_{i,j}$ added to each layer of a $J$-layer neural network is bounded by $N_c > 0$, denoted by:
\begin{equation}
    ||n_{i.j}|| \leq N_c.
\end{equation}
\end{assumption}

Given the bounded noise, the following corollary holds:
\begin{corollary}
    If the global model has $J$ layers, the additional noise $n_i^t$ in global round $t$ satisfies: 
    \begin{equation}
        ||n_i^t|| = ||\sum_{j=1}^J n_{i,j}^t|| \leq JN_c,
    \end{equation}
    where $n_{i,j} \sim N(0, \sigma_{i.j}^2)$ in \tool, $\sigma_{i,j}$ is defined in Equation~\ref{equ:adaptive_dp} and satisfying the constraint condition in Equation~\ref{thm:equ:my_dp}.
    \label{coro:1}
\end{corollary}

Subsequently, leveraging the bounded noise assumption, we derive a theorem to establish the convergence of \tool under $(\epsilon, \delta)$-privacy constraints. 
\begin{theorem}[Convergence of \tool]
Assume that Assumptions~\ref{Ass:2}-\ref{Ass:4} hold and let $L$, $\mu$, $G_c$, $N_c$ be defined therein. For a $J$-layer network, let the learning rate $\eta$ satisfy $\frac{2JN_c}{G_c} > \eta > \frac{2JN_c}{G_c} - \frac{\mu -L}{LG_c\mu}$, then we have the following convergence results:
\label{Thm:conv}
    \begin{equation}
        \begin{split}
             \quad \ \mathbb{E}[F(w^t_g)] - F^* & \leq 
            (\frac{L+\psi\mu  L}{\mu})^t(\mathbb{E}[F(w^0_g)] - F^*) \\ & +\frac{\psi + 2\phi}{2}\sum_{n=0}^{t-1} (\frac{L+\psi\mu L}{\mu})^n,
        \end{split}
    \end{equation}
    where $0 < \frac{L+\psi\mu  L}{\mu} < 1$, $\phi = \frac{L\eta^2}{2}G_c^2 + 2LJ^2N_c^2$, and  $\psi = 2JN_c - \eta G_c$.
\end{theorem}

\begin{proof}
   See Section~\ref{sec:appendixC} in the Supplemental Material.
\end{proof}

Theorem~\ref{Thm:conv} demonstrates that although \tool injects random noise into the training process to ensure privacy protection for local clients' updates, this injection does not compromise the model's convergence performance when the noise size is limited. That said, \tool preserves the convergence process, thereby guaranteeing the effectiveness of the entire FL training.

\color{black}

%% file: input/exp_ret.tex
\section{Experimental Results}
\label{sec:6}
In this section, we first describe our evaluation setup, including the FL framework, the client configurations, the datasets and models, and the attack settings. To evaluate \tool, we focus on answering the following research questions:

\textbf{RQ 1}: How is \tool's performance in terms of privacy budget, prediction accuracy, noise injection amount, and resource consumption?

\textbf{RQ 2}: Can \tool defend against real-world adversarial attacks? 

\textbf{RQ 3}:  How sensitive is \tool to privacy assessment?

All experiments were conducted using an Intel(R) Xeon(R) CPU E5-2620 v4 @ 2.20GHz processor and an NVIDIA GeForce GTX 1080 Ti graphics card. Only representative results are included in this paper.

\subsection{Experimental Setup}
\label{sec:6-1}


\heading{Federated Learning Framework.} We implement \tool on a synchronous FL system with 100 clients and one central server, simulating real-world deployment scenarios. The system architecture follows the threat model in Section~\ref{sec:threat}, with one Honest-but-Curious (HBC) client among non-malicious participants.

\heading{Data Distribution.} We evaluate under realistic non-IID conditions using \textit{Private Label Isolation}, where one target class (\eg ``dog'' in CIFAR-10) is excluded from the HBC client to simulate sensitive data protection scenarios.

\heading{Training Configuration.} All experiments use:
\begin{itemize}
	\item 400 global rounds with 10\% client activation per round.
	\item Local training: $E = 2$ epochs, batch size = 50, $\eta$ = 0.1.
	\item Differential privacy: $\delta=0.02$ (reciprocal of batch size).
	\item Gradient clipping: $G_c=20$ (validated in Section~\ref{subsec:utility}).
\end{itemize}

\heading{Datasets \& Models.} We conduct comprehensive evaluations using the following dataset-model combinations.
\begin{table}[h]
	\centering
	\caption{Dataset-Model Combinations}
	\label{tab:datasets}
	\begin{tabular}{c|c|c|c}
		\hline
		\textbf{Evaluation} & \textbf{Dataset} & \textbf{Classes} & \textbf{Model(s)} \\ \hline
		\multirow{2}{*}{Performance} & CIFAR-100 & 100 & ResNet-18, CNN \\
		& CIFAR-10 & 10 & ResNet-18, CNN \\ \hline
		\multirow{2}{*}{Attack defense} & CIFAR-10 & 10 & ResNet-18 \\
		& FEMNIST & 62 & CNN \\ \hline
	\end{tabular}
\end{table}

\heading{Attack Simulation.} We evaluate against state-of-the-art reconstruction attacks using:
\begin{itemize}
	\item ACGAN~\cite{odena2017conditional} with discriminator architecture matching the global model.
	\item Fréchet Inception Distance (FID) scores as a reconstruction quality metric (Detailed results are provided in Supplemental Material).
\end{itemize}

\subsection{Performance Evaluation (RQ 1)}
We conduct comprehensive experiments to evaluate \tool's effectiveness across five dimensions: (1) privacy budget impact, (2) model accuracy, (3) noise efficiency, (4) computational overhead, and (5) data heterogeneity impact. Our evaluation compares \tool against six baselines: vanilla FL (no DP)~\cite{mcmahan2017communication}, Full DP~\cite{9069945}, Time-Varying DP~\cite{yuan2023amplitude}, Sensitive DP~\cite{xue2023differentially}, DPA LDP~\cite{zhang2024dynamic}, and AdapLDP~\cite{10851366}. All experiments follow the setup in Section~\ref{sec:6-1}.


\begin{figure}[!tp]
	\centering
	\subfloat[ResNet-18 @ CIFAR-10]{\includegraphics[width=3in]{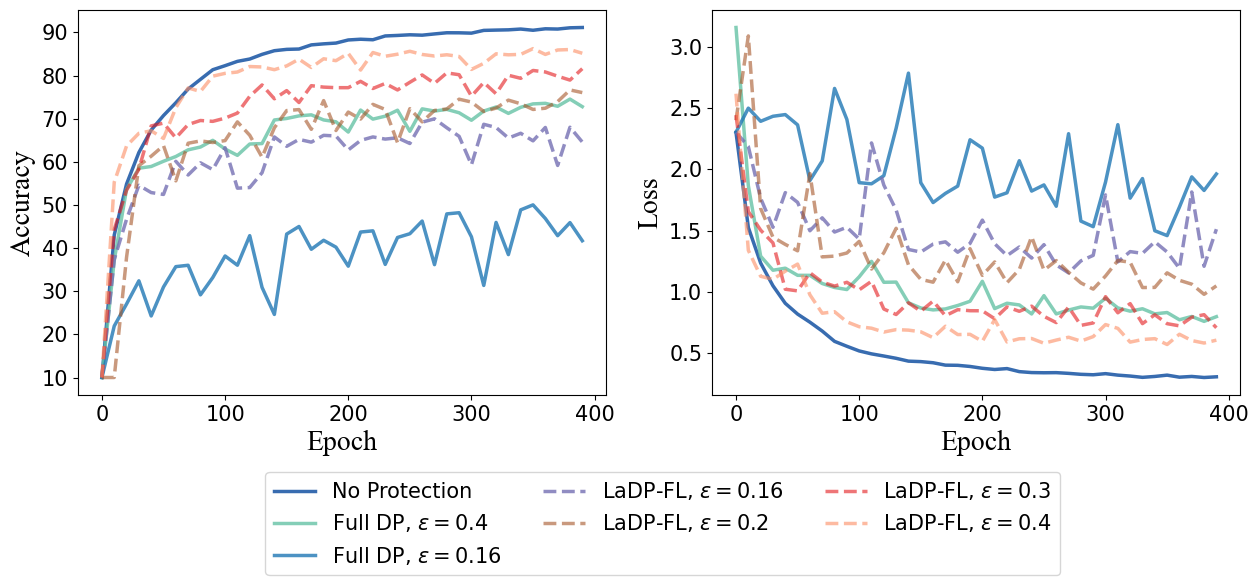}}
	\hspace{1mm}
	\subfloat[ResNet-18 @ CIFAR-100]{\includegraphics[width=3in]{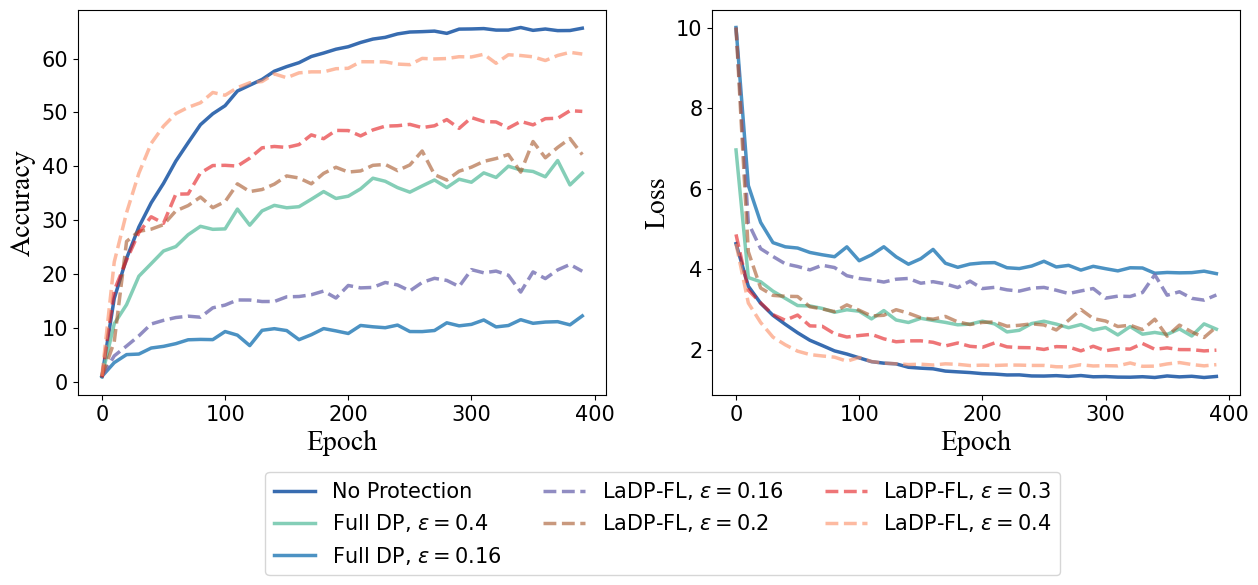}}
	\caption{The influence of the privacy budget $\epsilon$ on global accuracy and loss when using ResNet-18 on different datasets.}
	\label{fig_privacy_budget}
	\vspace{-2ex}
\end{figure}

\subsubsection{Impact of Privacy Budget}
\label{sec_impact_of_epsilon}
Figure~\ref{fig_privacy_budget} demonstrates \tool's robustness across varying privacy budgets ($\epsilon \in [0.16,0.4]$) on ResNet-18. We highlight two key findings: (1) \textbf{Accuracy Preservation:} At $\epsilon=0.16$, \tool maintains 22.4\% accuracy on CIFAR-100 vs. 12.23\% for Full DP (83.15\% improvement). As $\epsilon$ increases to 0.4, this gap narrows to 9.62\% (61.67\% vs. 52.05\%). (2) \textbf{Graceful Degradation:} When reducing $\epsilon$ from 0.4 to 0.16, \tool exhibits an average accuracy drop of 18.41\% (vs. 39.88\% for Full DP) across CIFAR-10/100 datasets, indicating more stable performance. This stability stems from \tool's layer-adaptive noise injection, which strategically prioritizes preserving critical model parameters.

\begin{table}[!t]
    \caption{Accuracy performance of \tool and baselines.}
    \centering
    \scriptsize
    \setlength{\tabcolsep}{1.5pt}
    \begin{tabular}{c|c|c|c c c c c c}
    \hline
    \multicolumn{3}{c|}{Method} & \makecell{\tool \\ (Ours)} & Full DP & \makecell{Time-\\Varying DP} & \makecell{Sensitive \\ DP} & DPA LDP & \makecell{AdapLDP} \\\hline
    \multicolumn{3}{c|}{Scenarios} & \multicolumn{6}{c}{\textbf{ResNet-18@CIFAR-10}}\\\hline
     \multirow{4}{*}{\makecell{Acc. \\ (\%)}} & \multirow{4}{*}{$\epsilon$} & \multirow{1}{*}{0.2} & \textbf{77.82} & 58.35& 63.65& 64.43 & 66.54 & 73.90\\
    & &  \multirow{1}{*}{0.3} & \textbf{82.29} & 58.58& 62.54& 67.64 & 73.77 & 77.12\\
    & &  \multirow{1}{*}{0.4} & \textbf{87.12} & 60.71& 66.79& 69.68 & 77.94 & 83.28\\
    & &  \multirow{1}{*}{0.5} & \textbf{90.06} & 70.59& 78.51& 84.68 &  79.36 & 88.39\\\hline
    \multicolumn{3}{c|}{\makecell{Average \\ Improvement Rate}} & - & 36.22\% & 24.73\% & 18.45\% & 13.44\%& 4.62\% \\\hline
    \multicolumn{3}{c|}{Scenarios} & \multicolumn{6}{c}{\textbf{CNN@CIFAR-100}}\\\hline
    \multirow{4}{*}{\makecell{Acc. \\ (\%)}} & \multirow{4}{*}{$\epsilon$} & \multirow{1}{*}{0.2} & \textbf{15.69}& 1.23 & 1.87 & 2.87 & 7.73 & 12.04\\
    & & \multirow{1}{*}{0.3} & \textbf{18.56}& 1.52& 2.03& 3.55 & 12.56 & 16.88\\
    & & \multirow{1}{*}{0.4} & \textbf{39.45}& 10.83& 22.44& 25.12 & 33.14 & 37.29\\
    & & \multirow{1}{*}{0.5} & \textbf{40.38}& 22.13& 28.93& 35.57 & 35.11& 39.26\\\hline
     \multicolumn{3}{c|}{\makecell{Average \\ Improvement Rate}} & -& 660.85\% & 417.18\% & 235.02\% & 46.20\%& 12.23\% \\\hline
    \end{tabular}
    \label{tab:acc_ret}
\end{table}

\subsubsection{Accuracy Comparison}
\label{subsec:utility}
Table~\ref{tab:acc_ret} shows the experimental results of \tool's prediction accuracy compared to SOTA DP mechanisms under identical privacy budgets ($\epsilon \in [0.2,0.5]$). We highlight three key findings: (1) \textbf{Low-Budget Superiority}: At $\epsilon=0.2$, \tool improves accuracy by 33.37\% over Full DP when training ResNet-18 on CIFAR-10, outperforming Time-Varying DP (22.26\%) and Sensitive DP (20.78\%). The advantage is more pronounced when training CNN on CIFAR-100 with 498.93\% average improvement. (2) \textbf{High-Budget Convergence}: When $\epsilon=0.5$, \tool approaches non-DP performance on CIFAR-10, while other methods show significant gaps (up to 20.14\% degradation). (3) \textbf{Consistent Gains}: Across all scenarios, \tool achieves average improvements of: 236.32\% over Full DP, 144.93\% over Time-Varying DP, 102.42\% over Sensitive DP, 25.18\% over DPA LDP, and 6.1\% over AdapLDP.

\heading{Technical Insight.} Existing approaches treat neural networks monolithically, while \tool carefully identifies and assesses the privacy capabilities of neural network layers that contribute more significantly to the final model's prediction ability. This prevents excessive noise injection into layers that do not compromise privacy but provide stronger predictive ability. Complete results and statistical significance tests are provided in Section~\ref{app:more_results} in the Supplemental Material.

\heading{Privacy-utility Tradeoff.} Overall, \tool achieves SOTA privacy-utility tradeoffs, offering stronger privacy guarantees while maintaining superior utility. In empirical evaluations on CIFAR-10 with ResNet-18, \tool outperforms Full-DP by reducing injected noise by 71.96\% and improving model accuracy by 36.23\%. Compared to Time-Varying DP and Sensitive DP, it reduces noise injection by 80.3\% and 9.4\% across all scenarios, respectively, while improving accuracy by 21.4\% and 15.9\%. Under identical privacy budgets, on average, \tool surpasses DPA LDP and AdapLDP with 25.18\% and 6.1\% higher accuracy, respectively, and reduces noise volume by 40.54\% and 68.31\%.

From a privacy perspective, \tool significantly curtails cumulative privacy budget consumption during training, achieving reductions of 63.3\% (vs. Full-DP), 76.7\% (vs. Time-Varying DP), 26.6\% (vs. Sensitive DP), 71.6\% (vs. DPA LDP), and 63.3\% (vs. AdapLDP). These results demonstrate that \tool: (1) optimizes noise injection via fine-grained calibration, minimizing utility degradation; and (2) rigorously controls privacy budget expenditure, enhancing end-to-end protection guarantees—a critical advancement for differentially private ML in practice.

\begin{table}[!t]
    \caption{Noise scale performance of \tool and baselines.}
    \centering
    \scriptsize
    \setlength{\tabcolsep}{1.5pt}
    \begin{tabular}{c|c|c|c c c c c c}
    \hline
    \multicolumn{3}{c|}{Method} & \makecell{\tool \\ (Ours)} & Full DP & \makecell{Time-\\Varying DP} & \makecell{Sensitive \\ DP} & DPA LDP & \makecell{AdapLDP}\\\hline
    \multicolumn{3}{c|}{Scenarios} & \multicolumn{6}{c}{\textbf{ResNet-18@CIFAR-10}}\\\hline
    \multirow{4}{*}{Noise Scale} & \multirow{4}{*}{$\epsilon$} & \multirow{1}{*}{0.2} &  275,447 & 727,901& 862,998 & \textbf{242,583}& 386,453 & 735,138\\
    & & \multirow{1}{*}{0.3} & \textbf{132,365}& 505,266& 295,231& 165,234 & 261,698 & 512,366 \\
    & & \multirow{1}{*}{0.4} & \textbf{87,563} & 383,945& 224,129& 95,468 & 241,778 & 402,970 \\
    & & \multirow{1}{*}{0.5} & \textbf{51,658} & 203,968& 193,487& 53,997 & 156,947 & 217,344\\\hline
    \multicolumn{3}{c|}{\makecell{Average \\ Reduction Rate}} & - & 71.96\% & 64.37\% & 4.74\% & 52.25\% & 72.80\%\\\hline
    \multicolumn{3}{c|}{Scenarios}& \multicolumn{6}{c}{\textbf{CNN@CIFAR-100}}\\\hline
    \multirow{4}{*}{Noise Scale} & \multirow{4}{*}{$\epsilon$} & \multirow{1}{*}{0.2} & \textbf{22,968} & 77,523 & 50,803& 23,854 & 35,603 & 72,781\\
    & & \multirow{1}{*}{0.3} & \textbf{19,297} & 52,652 & 34,114& 21,271 & 37,329 & 56,700\\
    & & \multirow{1}{*}{0.4} & \textbf{9,028} & 42,236 & 23,578& 9,548 & 27,668 & 44,469\\
    & & \multirow{1}{*}{0.5} & \textbf{8,677} & 25,786 & 18,270& 8,933 & 10,119 & 18,574\\\hline
    \multicolumn{3}{c|}{\makecell{Average \\ Reduction Rate}} & -& 69.68\% & 53.11\% & 5.33\% & 41.35\% & 66.85\% \\\hline
    \end{tabular}
    \label{tab:noise_ret}
\end{table}

\subsubsection{Noise Scale Comparison}
\label{subsec:privacy}
We quantify noise using the $\mathcal{L}_2$ norm, the de facto standard for perturbation measurement in adversarial learning~\cite{athalye2018obfuscated,madry2017towards}. For each privacy mechanism, we compute the cumulative noise injected over 400 global training rounds. Table~\ref{tab:noise_ret} summarizes the results.

\heading{Key Findings.} For ResNet-18/CIFAR-10 at $\epsilon=0.5$, \tool reduces noise by 74.67\% compared to Full DP~\cite{9069945}, outperforming Time-Varying DP (5.14\% reduction) and Sensitive DP (73.53\% reduction). DPA LDP achieves a 23.05\% average reduction, while AdapLDP matches Full DP’s noise volume. At stricter privacy ($\epsilon=0.2$), \tool maintains an average 41.59\% reduction while achieving higher accuracy than baselines, demonstrating superior privacy-utility tradeoffs.

We evaluate \tool against all DP variants, averaging results across all scenarios (see Section~\ref{app:more_results} in the Supplemental Material for methodology). Compared to: (1) Full DP/AdapLDP: 69.32\%/68.31\% reduction, (2) Time-Varying DP: 51.35\% reduction, (3) Sensitive DP: 1.2\% reduction (marginal), and (4) DPA LDP: 40.54\% reduction. Overall, \tool achieves a 46.14\% average noise reduction versus all SOTA methods.

\heading{Technical Insight.} Full DP and AdapLDP inject static noise, while Time-Varying DP and DPA LDP use exponential decay to reduce final noise. Sensitive DP adjusts noise per iteration via sensitivity estimation. Conversely, \tool adopts a layer-aware approach: (1) Dynamic per-layer allocation: Noise scales with layer-specific requirements. (2) Client selection variance: Fluctuations arise from differing per-round client contributions, proving adaptability to evolving privacy needs.

\begin{figure}[!tp]
\centering
\subfloat[$\epsilon = 0.2$]{\includegraphics[width=1.5in]{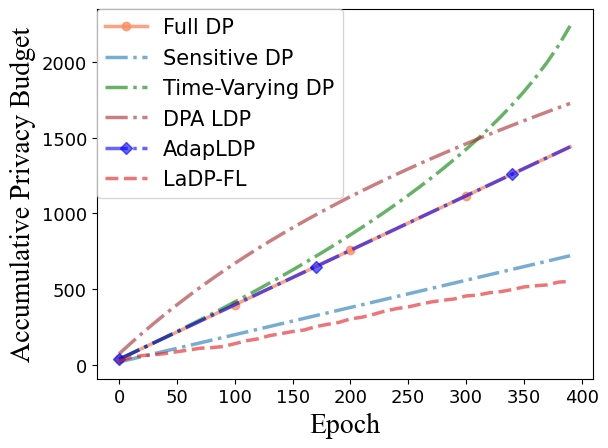}}
\hspace{1mm}
\subfloat[$\epsilon = 0.5$]{\includegraphics[width=1.5in]{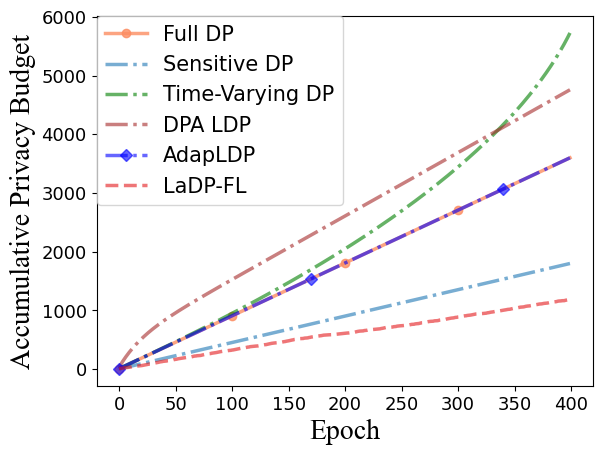}}
\caption{The comparison of accumulative privacy budget under different privacy budgets $\epsilon$ between \tool and baselines when training ResNet-18 on the CIFAR-10 dataset. Notably, the curves for Full DP and AdapLDP completely overlap; we distinguish them using distinct markers at varying intervals.}
\label{fig_privacy_account}
\vspace{-2ex}
\end{figure}

\subsubsection{Accumulative Privacy Budget Comparison}
\label{sec:accumulate_privacy_budget}
To rigorously quantify the privacy-utility trade-off, we evaluate the cumulative privacy budgets of \tool and baselines by Naive Composition Theorem~\cite{dwork2014algorithmic} (detailed derivations in Section~\ref{app:theoretical_background} of the Supplemental Material). Figure~\ref{fig_privacy_account} presents the results when training ResNet-18 on CIFAR-10 ($\delta=0.02$).

\heading{Key Findings.} Full DP demonstrates linear budget growth with training rounds (\eg slope=3.6 when training ResNet-18 under $\epsilon=0.2$), as expected from fixed-noise injection. Time-Varying DP and DPA LDP exhibit 57.3\% and 26.04\% higher cumulative $\epsilon$ than Full DP, respectively, confirming their privacy relaxations for utility gains. AdapLDP maintains an identical $\epsilon$ progression to Full DP, as both employ equivalent noise mechanisms. Sensitive DP reduces $\epsilon$ accumulation by 51.2\%, demonstrating its adaptive protection approach.

\tool achieves 63.3\% lower $\epsilon$ accumulation than Full DP through our \textit{layer-wise selective injection} and \textit{adaptive variance scaling} strategies. Crucially, this improvement is attained while maintaining model utility, as \tool's noise allocation is optimized via our proposed privacy-aware distribution mechanism. The results are statistically significant ($p=$2.64e-5 under Mann-Whitney U Test) and consistent across all four experimental scenarios (Section~\ref{app:more_results} in the Supplemental Material).

\heading{Security Implications.} The reduced $\epsilon$ accumulation of \tool indicates a stronger capability of privacy preservation. Notably, although \tool decreases the amount of noise injection, it achieves this by adaptively scaling or shrinking the variance for different layers and targeting noise injection at fewer, yet impactful layers.

\subsubsection{Resource Comparison}
\label{sec:resource_compare}

\begin{table}[!t]
    \caption{Resource utilization of \tool and baselines.}
    \centering
    \scriptsize
    \setlength{\tabcolsep}{1.5pt}
    \begin{tabular}{c|c|c|c c c c c c}
    \hline    
    \multicolumn{3}{c|}{Method} & \makecell{\tool \\ (Ours)} & Full DP & \makecell{Time-\\Varying DP} & \makecell{Sensitive \\ DP} & DPA LDP & \makecell{AdapLDP}\\\hline
    \multicolumn{3}{c|}{Scenario} & \multicolumn{6}{c}{\textbf{ResNet-18@CIFAR-10}}\\\hline
      \multirow{4}{*}{\makecell{Training \\ Duration (s)}} & \multirow{4}{*}{$\epsilon$} & \multirow{1}{*}{0.2} &  93,775 & 80,110& 86,423& 83,793 & 91,687 & 94,773\\
    & & \multirow{1}{*}{0.3} & 93,455 & 82,157 & 83,151 & 86,363 & 89,569 & 92,042 \\
    & & \multirow{1}{*}{0.4} & 94,985 & 81,511 & 87,562 & 82,443 & 90,524 & 96,111 \\
    & & \multirow{1}{*}{0.5} & 96,008 & 82,018 & 89,110 & 84,008 & 92,272 & 95,647 \\\hline
    \multicolumn{3}{c|}{\makecell{Average \\ Incremental Rate}} & - & 16.10\% & 9.28\% & 12.41\% & 3.90\% & -0.08\% \\\hline
    \multicolumn{3}{c|}{Scenario}& \multicolumn{6}{c}{\textbf{CNN@CIFAR-100}}\\\hline
      \multirow{4}{*}{\makecell{Training \\ Duration (s)}} & \multirow{4}{*}{$\epsilon$} & \multirow{1}{*}{0.2} & 16,731 & 15,120 & 16,368 & 16,081 & 16,441 & 16,934\\
    & & \multirow{1}{*}{0.3} & 16,795 & 15,236 & 15,831 & 16,092 & 16,398 & 16,765\\
    & & \multirow{1}{*}{0.4} & 16,963 & 15,045 & 15,713 & 16,045 & 16,613 & 16,782\\
    & & \multirow{1}{*}{0.5} & 16,621 & 15,013 & 15,963 & 16,140 & 16,286 & 16,980 \\\hline
    \multicolumn{3}{c|}{\makecell{Average \\ Incremental Rate}}& - & 11.09\% & 5.10\% & 4.28\% & 2.09\% & -0.51\%\\\hline
    \end{tabular}
    \label{Tab:resource_perform}
\end{table}

Our resource evaluation demonstrates two key findings: 

\heading{Key Findings.} As shown in Table~\ref{Tab:resource_perform}, compared to Full DP, \tool introduces an average of 13.56\% additional training time across scenarios. This overhead reduces to 3\%-8\% when compared to adaptive methods (Time-Varying DP, Sensitive DP, and DPA LDP), and becomes comparable to AdapLDP.

\heading{Practical Implications.} The absolute overhead (max +17.06\%) translates to $<$4 additional hours for ResNet-18 training - a reasonable tradeoff for the achieved privacy benefits (Section~\ref{subsec:utility}). We note this overhead would diminish proportionally with hardware improvements, as the computational complexity remains within $O(J)$ of baseline methods, where $J$ denotes the depth of the neural network.


\subsubsection{\mapping{Comment 1.7 in R1 \& Comment 1.3 \& 1.6 in R2}Impact of Data Heterogeneity}
\label{sec:impact_of_noniid}

Our previous evaluation (Sections~\ref{sec_impact_of_epsilon}--\ref{sec:resource_compare}) used a non-IID distribution (Section~\ref{sec:6-1}) where non-malicious clients had complete label coverage while only HBC clients lacked the private label. To evaluate \tool's resilience to data heterogeneity, we conduct experiments with ResNet-18 on CIFAR-10 under two pathological non-IID distributions: (1) \textbf{Distribution I (Label Scarcity):} 99 non-malicious clients have data with 4 labels (including the private label), while 1 HBC client has 3 labels (excluding the private label). All clients maintain equal data quantities. (2) \textbf{Distribution II (Quantity Skew):} Non-private data follows a Hetero-Dirichlet Distribution~\cite{JMLR:v24:22-0440} ($\alpha=0.01$) across all clients, while private data follows another Hetero-Dirichlet distribution ($\alpha=0.01$) among the 99 non-malicious clients. Figure~\ref{fig_non_iid} visualizes these distributions by showing data allocations for 4 randomly selected non-malicious clients and 1 HBC client, with Label 10 representing private data.

\begin{figure}[!htp]
	\centering
	\subfloat{\includegraphics[width=0.9\linewidth]{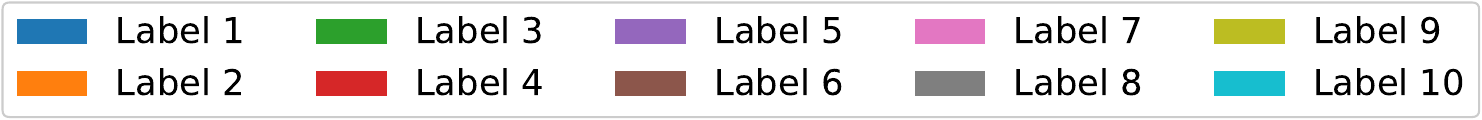}}
	\setcounter{subfigure}{0}
	\subfloat[General Distribution]{\includegraphics[width=0.32\linewidth]{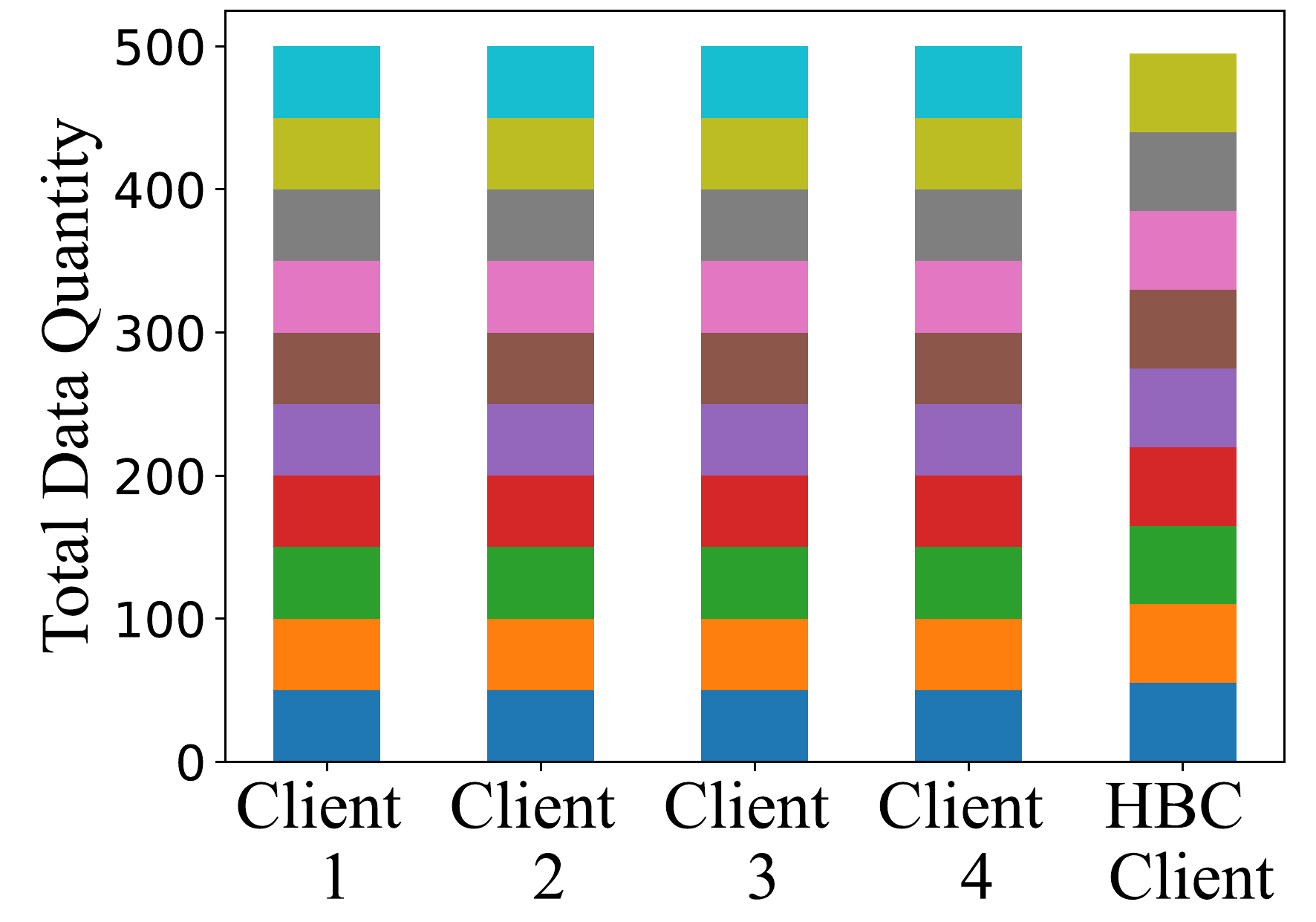}\label{fig:gen-distr}}
	\subfloat[Distribution I]{\includegraphics[width=0.32\linewidth]{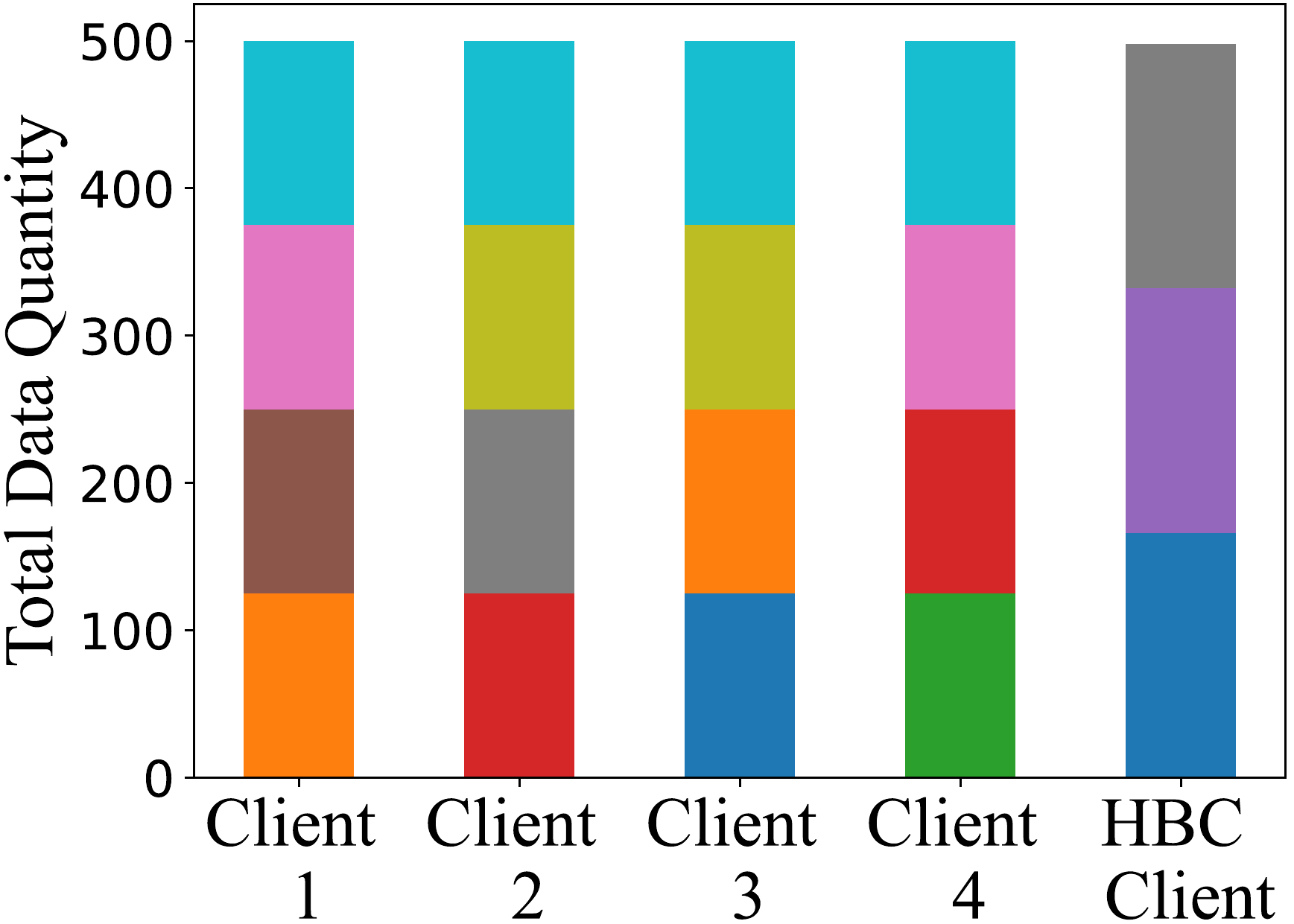}}
	\subfloat[Distribution II]{\includegraphics[width=0.32\linewidth]{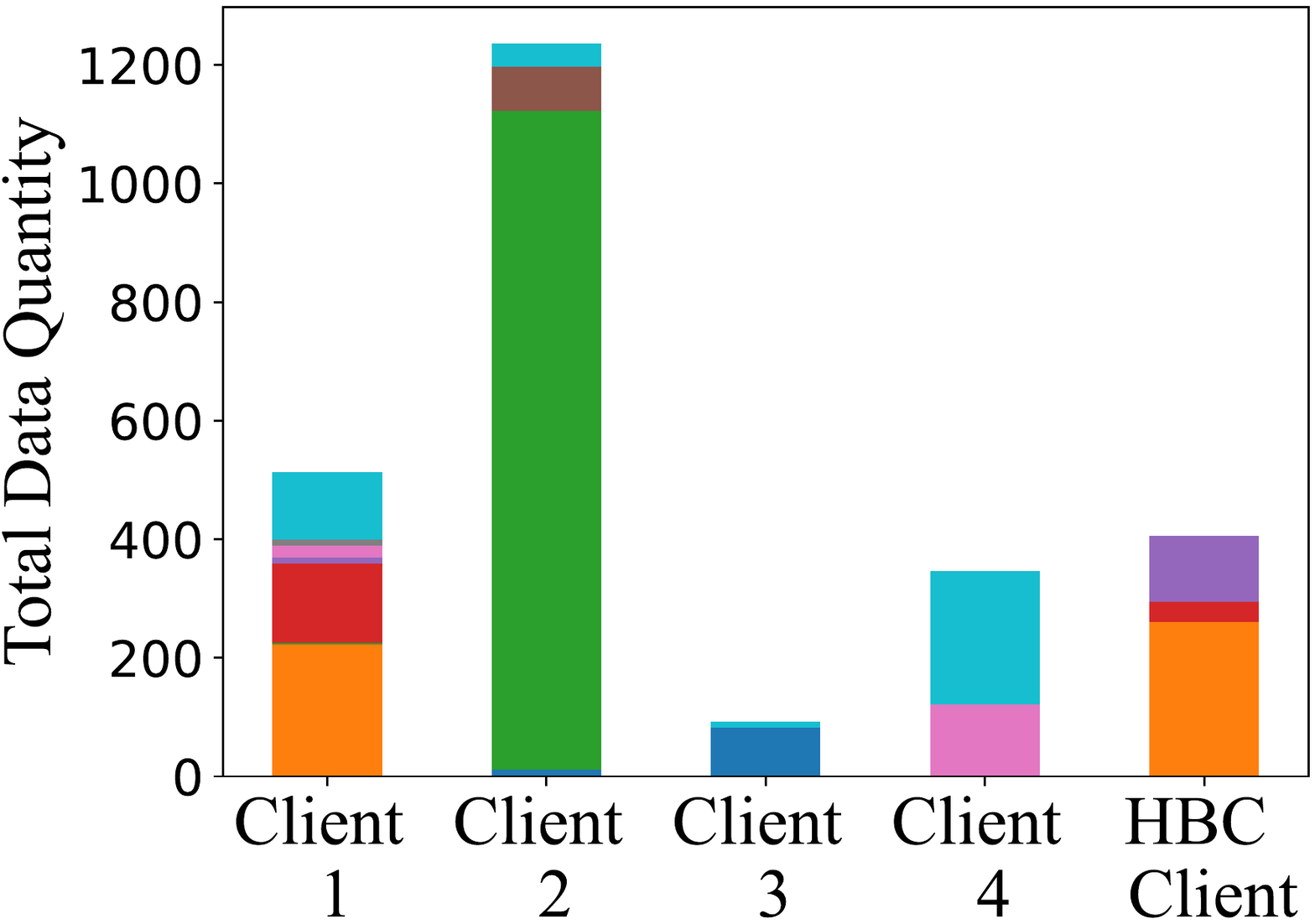}}
	\caption{Data distributions across clients under different heterogeneity settings. Label 10 denotes private data.}
	\label{fig_non_iid}
\end{figure}

\heading{Key Findings.}	(1) \textbf{Accuracy Superiority:} \tool maintains significant accuracy advantages: 38.1\% over Full DP, 25.5\% over Time-Varying DP (Distribution I), and 80.5\% over Full DP, 51.5\% over Time-Varying DP (Distribution II). (2) \textbf{Noise Efficiency:} Despite distribution skewness, \tool achieves 3.09\% average noise reduction vs. Sensitive DP and minimal noise inflation compared to fixed-schedule baselines. (3) \textbf{Robustness:} No failure cases observed in privacy-utility tradeoffs, demonstrating consistent superiority across all tested scenarios.

\heading{Technical Insight.} The results stem from \tool's dynamic noise adaptation through its Layer Selection and Privacy Estimation Modules, contrasting with the static approaches of Full DP (727,894 noise scale) and Time-Varying DP (356,555 noise scale). While extreme heterogeneity necessitates slight noise increases (251,207 $\rightarrow$ 274,003), \tool's adaptive mechanism preserves better utility than sensitivity-based alternatives like Sensitive DP (253,336 $\rightarrow$ 283,316).

\begin{table}[!t]
	\centering
	\caption{Performance under extreme data heterogeneity ($\epsilon \in \{0.2,0.5\}$).}
	\scriptsize
	\begin{tabular}{c|l|c|c|c|c}
		\hline
		\multirow{2}{*}{$\epsilon$} & \multirow{2}{*}{Method} & \multicolumn{2}{c|}{Distribution I} & \multicolumn{2}{c}{Distribution II} \\ \cline{3-6}
		&  & Acc. (\%)& Noise Scale & Acc. (\%) & Noise Scale \\ \hline
		- & No Protection & 61.33 & - & 45.16 & - \\ \hline
		\multirow{6}{*}{0.2} & Full DP & 30.32 & 727,894 & 15.13 & 727,895\\
		& Time-Varying DP & 34.44 & 356,555 & 17.77 & 356,558\\
		& Sensitive DP & 38.76 & 253,336 & 21.39 & 283,316 \\
		& DPA LDP & 39.07 & 398,653 & 24.95 & 398,153 \\
		& AdapLDP & 35.24 & 719,643 & 19.1 & 700,645 \\
		& \textbf{LaDP-FL (Ours)} & \textbf{42.13} & \textbf{251,207} & \textbf{26.86} & \textbf{274,003}\\ \hline
		\multirow{6}{*}{0.5} & Full DP & 35.49 & 203,955 & 17.33 & 211,161  \\
		& Time-Varying DP & 37.88 & 178,278 & 20.94 & 162,628\\
		& Sensitive DP & 40.11 & 87,669 & 26.03 & 112,989\\
		& DPA LDP & 43.77 & 150,691 & 27.91 & 140,293 \\
		& AdapLDP & 46.62 & 197,558 & 29.52 & 229,980 \\
		& \textbf{LaDP-FL (Ours)} & \textbf{48.73} & \textbf{84,013} & \textbf{31.79} & \textbf{108,364}\\\hline
	\end{tabular}
	\label{tab:heterogeneous}
\end{table}

\subsection{Defense Performance (RQ 2)}
\label{sec:defense-eval}
We evaluate \tool's defense against targeted class-level reconstruction attacks (Section~\ref{sec:6-1}), where adversaries aim to recover private-class data (\eg class `5'). Unlike pixel-wise metrics (PSNR/SSIM), we focus on distributional similarity using the Fréchet Inception Distance (FID), which quantifies the Wasserstein-2 distance between real and reconstructed private-class distributions in Inception-v3 feature space.

\begin{figure}[!tp]
\centering
\includegraphics[width=3in]{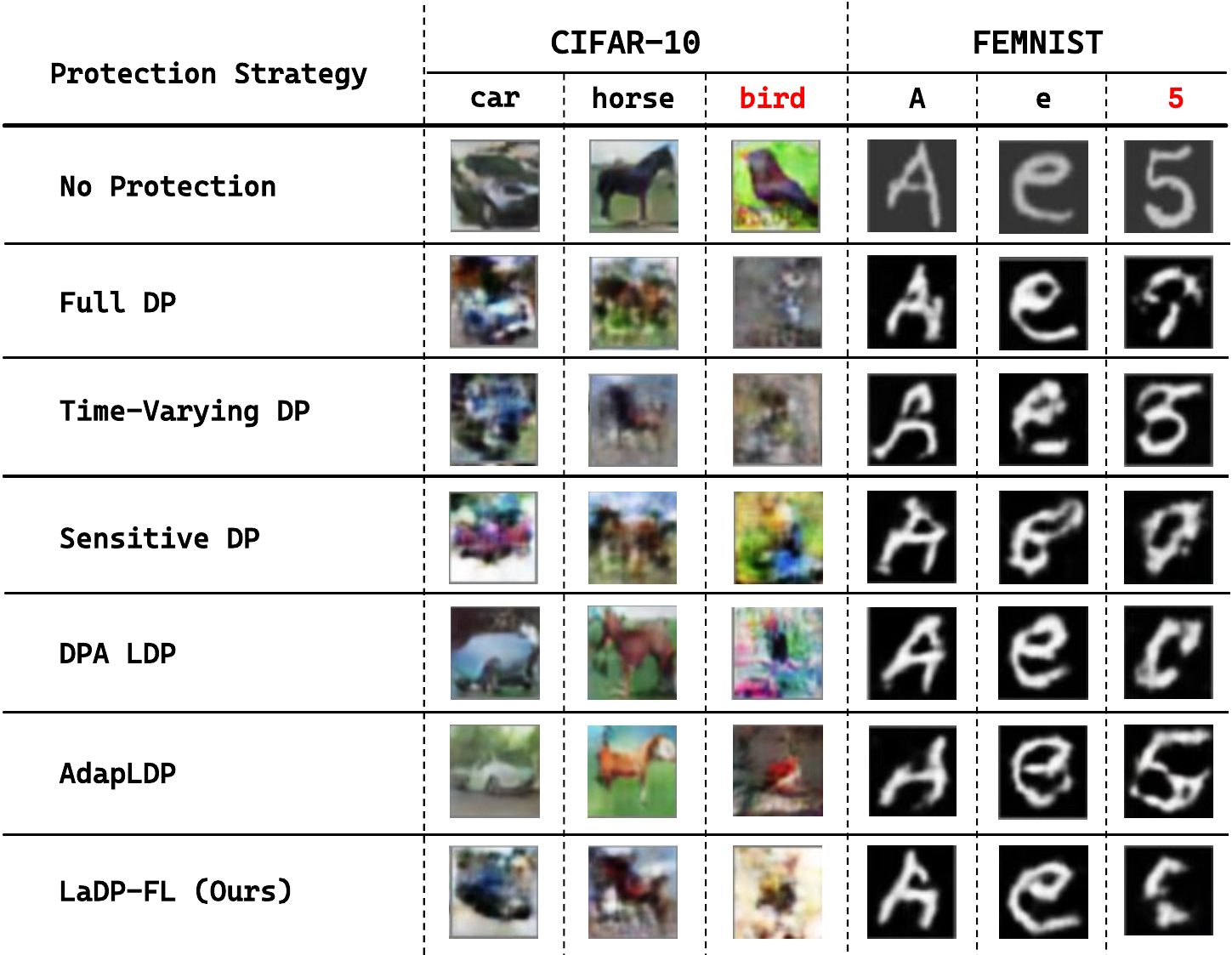}
\caption{The defensive capabilities of different privacy defense strategies against private data reconstruction attacks, where the black label indicates the class data shared by all clients and the red label indicates the private data only possessed by the victim clients.}
\label{fig_def}
\end{figure}

\heading{Visual Reconstruction Analysis.}
Figure~\ref{fig_def} shows the reconstruction results, where black-labeled classes represent public data shared across all clients, and red-labeled classes indicate private data exclusive to victim clients. Among all protection strategies, \tool provides the strongest protection, as the images reconstructed by the HBC client based on \tool are the most distorted from their counterparts. \tool achieves an FID of 121.02 and 83.13 on the privacy data within CIFAR-10 and FEMNIST, respectively (vs. 52.77 and 24.35 for no protection), indicating an average of 17.58\% greater protection than Full DP (FID=110.65 and 66.1). 

\subsection{Necessity of Privacy Estimation (RQ 3)}
To rigorously evaluate the impact of KL-divergence-based privacy estimation in \tool, we conduct a controlled ablation study comparing two variants: (1) \tool with privacy estimation enabled and (2) \tool with privacy estimation disabled. We measure their performance across two key metrics: model utility (test accuracy) and noise injection dynamics, as illustrated in Figure~\ref{fig_kl_important}.

\heading{Key Findings.} (1) \textbf{Utility Preservation:} Disabling privacy estimation degrades model accuracy by 15.97\% (high noise, $\epsilon = 0.2$) and 4.74\% (low noise, $\epsilon = 0.5$), demonstrating that adaptive noise calibration is critical for balancing privacy-utility trade-offs. (2) \textbf{Noise Adaptation:} Without privacy estimation, noise injection remains relatively static across layers and clients. In contrast, \tool’s privacy-aware optimization dynamically tunes noise magnitudes per layer and client, reducing redundancy (\eg lowering noise in non-sensitive layers) while enforcing stronger protection where needed. This adaptability improves privacy guarantees without sacrificing accuracy.

Our results validate that KL-divergence-based estimation is not merely auxiliary but fundamental to \tool's design. The module’s ability to quantify and react to privacy risks in real time distinguishes it from static DP-FL variants, enabling finer-grained noise allocation.

\begin{figure}[!tp]
	\centering
        \subfloat{\includegraphics[width=0.9\linewidth]{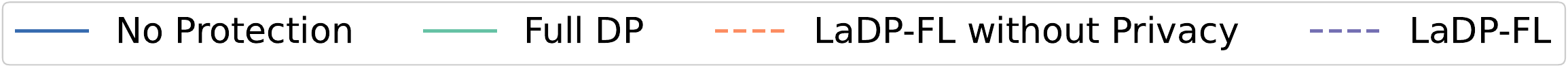}}
	\setcounter{subfigure}{0}
	\subfloat[Accuracy]{\includegraphics[width=0.3\linewidth]{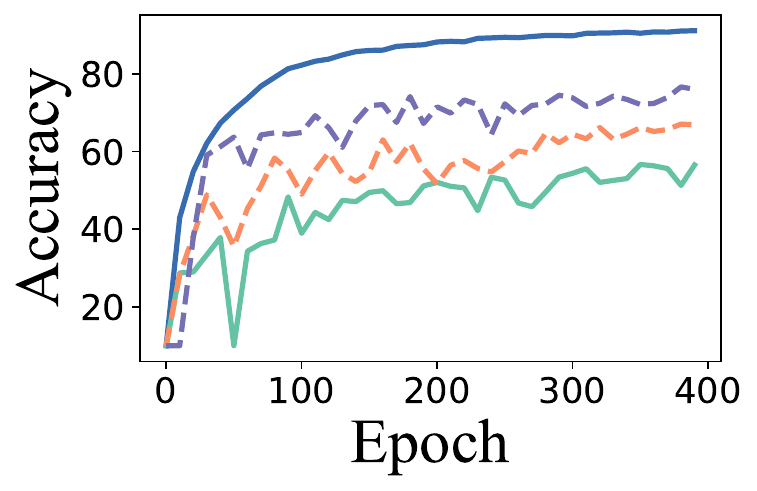}}
        \subfloat[Loss]{\includegraphics[width=0.3\linewidth]{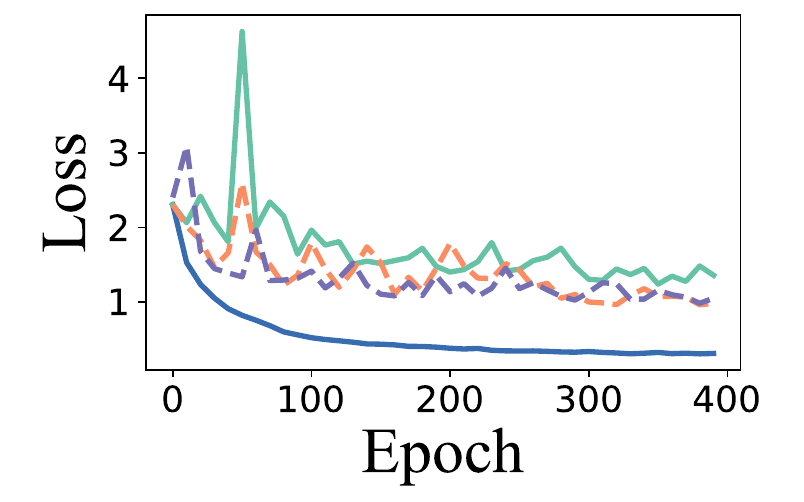}}
	\subfloat[Noise scale]{\includegraphics[width=0.3\linewidth]{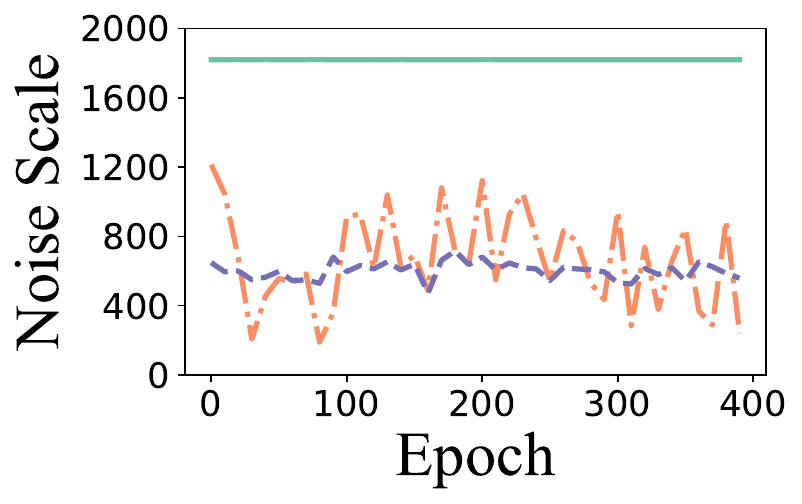}}
	\hspace{2mm}
	\hspace{2mm}
	\subfloat[Accuracy]{\includegraphics[width=0.3\linewidth]{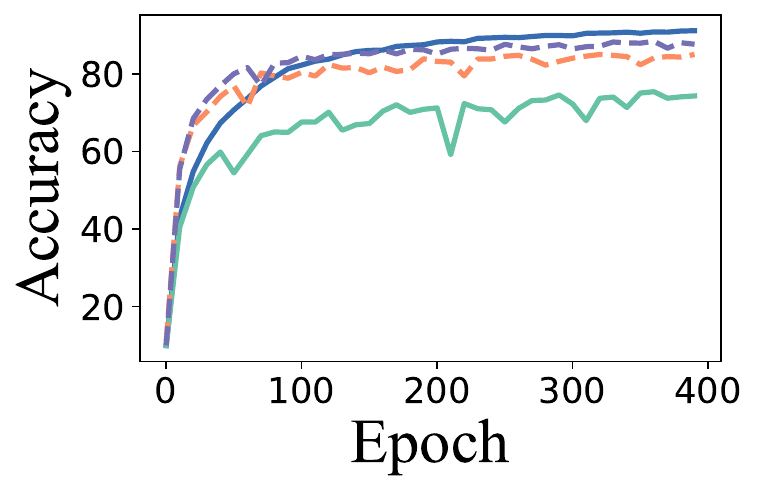}}
        \subfloat[Loss]{\includegraphics[width=0.3\linewidth]{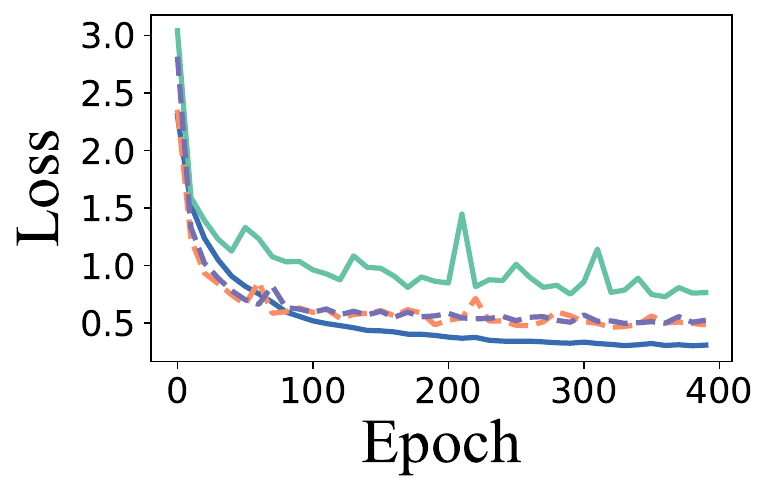}}
	\subfloat[Noise scale]{\includegraphics[width=0.3\linewidth]{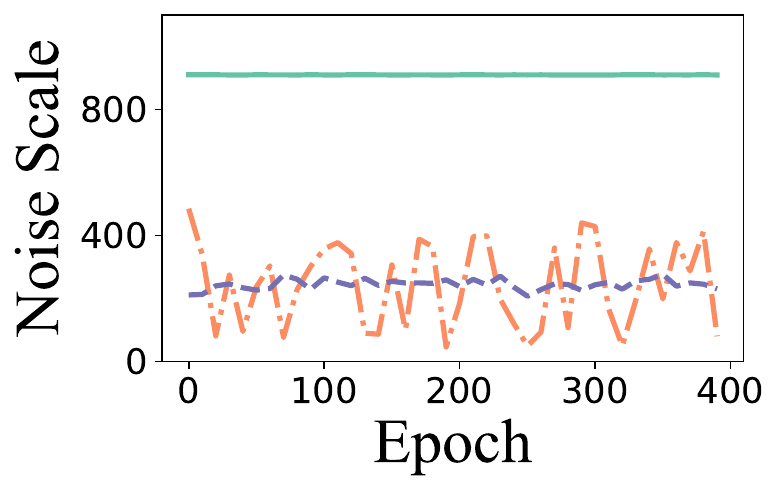}}
	\caption{The impact of privacy assessment component on global accuracy, loss function, and noise injection under $\epsilon=0.2$ ((a)-(c)) and $\epsilon=0.5$ ((d)-(f)).}
	\label{fig_kl_important}
\end{figure}



%% file: input/discussion.tex
\section{Discussion}

\subsection{\mapping{Comment 2.1 in R1}Implementation Complexity and Layer Selection}
\label{subsec:complexity}
While selecting critical layers for noise injection introduces additional computations compared to full-model DP approaches, \tool's complexity remains linear with respect to network depth (\ie $O(J)$ for $J$ layers). This is identical to the complexity of other DP-based FL approaches, as those algorithms must also iterate through all model layers and generate noise matrices matching each layer's parameter size for perturbation. In contrast, after processing all layers through its layer selection module, \tool only perturbs a small number of selected layers, thereby maintaining this linear complexity. As shown in Table~\ref{Tab:resource_perform}, these optimizations result in only 10.23\%-17.05\% increased computation time compared to Full DP, while providing superior privacy-utility tradeoffs. For very deep networks, our experiments on ResNet-152 (see results in the Supplemental Material) demonstrate that \tool achieves 53.04\% of the accuracy improvement rate while only introducing overhead by 7.27\% compared to Full DP.

\subsection{\mapping{Comment 1.7 in R1 \& Comment 2.4 in R1 \& Comment 1.3 in R2}Scalability in Large-Scale Federated Learning}
\label{subsec:scalability}

\heading{Deep Model Scalability:} As discussed in Section~\ref{subsec:complexity}, \tool achieves the same complexity, $O(J)$, as other DP-based FL algorithms. When deploying \tool on large models, all algorithms incur additional computational overhead. While advanced hardware resources (\eg more powerful GPUs) or optimization algorithms (\eg LoRA~\cite{hu2022lora}) can mitigate this computational burden, addressing this challenge is beyond the scope of this paper.

\heading{Multi-Client Considerations:} The non-IID robustness analysis (Sections~\ref{module3} and~\ref{sec:impact_of_noniid}) confirms \tool's effectiveness across heterogeneous clients. Meanwhile, \tool does not require any collaboration between clients, as all operations are performed independently by each client, ensuring scalability with respect to the number of clients.

\heading{Communication Efficiency:} Compared to Full DP (the fundamental DP-based FL algorithm), \tool requires no additional information transmission during communication, maintaining high efficiency. Furthermore, because \tool imposes no communication-specific constraints, it remains entirely orthogonal to existing advanced FL communication optimization algorithms~\cite{chai2020fedat,wu2022kafl,zhang2023fedmds} and can be seamlessly integrated with them.

\subsection{\mapping{Comment 1.6 in R1}KL Divergence and Noise Injection Sensitivity}
\label{subsec:kl_relationship}

The KL divergence ($P_{i,j}$) quantifies the information distance between local and global layer distributions, serving as our privacy risk indicator. The noise injection process follows Equation~\ref{equ:adaptive_dp}. In this equation, when a small change $\Delta P$ in $P_{i,j}$ is introduced, we have $\sigma_{i,j} + \Delta \sigma =  \frac{c_i \Delta f_{i,j}}{\epsilon (P_{i,j}+\Delta P)}$. By substituting $\sigma_{i,j}$ from Equation~\ref{equ:adaptive_dp} into the equation, we obtain $\Delta \sigma = -\frac{c_i \Delta f_{i,j}}{\epsilon P_{i,j}(P_{i,j}/\Delta P+1)}$. Given that $c_i, \Delta f_{i,j}, \epsilon, P_{i,j} \geq 0$, we derive $|\Delta \sigma| = \frac{c_i \Delta f_{i,j}}{\epsilon P_{i,j}(P_{i,j}/|\Delta P|+1)}$.

This design ensures: (1) \textbf{Proportional Protection:} noise injection exhibits a reciprocal square relationship with the KL divergence (indicating the level of private information). (2) \textbf{Stability:} The reciprocal relationship prevents noise explosion when $P_{i,j}\rightarrow 0$ through our clipping bound $B$. (3) \textbf{Adaptivity:} The $\Delta\sigma/\Delta P$ sensitivity derived above automatically adjusts noise scales during training. As shown in Figure~\ref{fig_kl_important}, this mechanism reduces unnecessary noise by 69.32\% compared to Full DP while maintaining privacy guarantees.

\subsection{\mapping{Comment 1.3 in R2}Practical Deployment Considerations}
\tool's modular design enables integration with existing FL frameworks like TensorFlow Federated and PySyft, as well as compatibility with common operating systems such as Windows and Ubuntu, facilitating deployment across most device environments. However, it still faces deployment challenges inherent to all FL systems, such as heterogeneous battery resources and unstable network conditions. Notably, our primary objective is to enhance the privacy-utility trade-off in privacy-sensitive scenarios, while addressing these specific deployment limitations falls outside the scope of this work.

%% file: input/summ.tex
\section{Conclusion}
\label{sec:7}

We present \tool, the first layer-wise adaptive noise injection framework for FL that achieves an optimal privacy-utility tradeoff. By selectively perturbing critical layers based on their privacy level (measured via KL divergence) and contribution to model utility, \tool significantly reduces unnecessary noise injection while maintaining strong DP guarantees. Theoretical analysis confirms \tool's convergence and asymptotic optimality under standard assumptions.

Empirical results demonstrate that \tool outperforms SOTA DP-based FL methods across multiple metrics: (1) \textbf{Accuracy:} Achieves 102.99\% higher accuracy on average than baselines (\eg 236.32\% over Full DP). (2) \textbf{Efficiency:} Reduces noise injection by 46.14\% while mitigating privacy budget accumulation by 60.3\% on average. (3) \textbf{Robustness:} Effectively thwarts reconstruction attacks, with HBC clients failing to recover private data (\eg $>$12.84\% increase in the FID of the reconstructed private data.).

\tool's design is agnostic to model architectures and scales efficiently with the number of clients. Future work includes extending \tool to asynchronous FL settings and further optimizing its overhead for large-scale models.
\color{black}

%% file: input/appendix.tex
\clearpage
$$\textbf{\LARGE Supplemental Material}$$
\section{Proof of Theorem~\ref{theorem_sensitivity}} 
\label{sec:appendixA}

Due to the page limit, we only provide the proof sketch here. More details can be found in the history-revised version.

\textbf{Proof Sketch of Theorem~\ref{theorem_sensitivity}:}
Based on existing works~\cite{zhou2022pflf}, the evaluation function $f(\cdot)$ under Differential Privacy (DP) theory during neural network training is generally defined as the final model output on the corresponding training dataset. To evaluate the sensitivity within the DP framework, we need to compute the difference between the outputs of the evaluation function on two adjacent datasets and take the maximum norm of this difference across all such pairs. At this point, performing one gradient descent step on both evaluation functions $f(x)$ and $f(x')$ and subsequently applying the triangle inequality allows us to separate the gradient term from the model term. This yields the difference between the final model outputs on the two adjacent datasets at the previous epoch and the difference between their gradients at the previous epoch. For the gradient term, Assumption~\ref{Ass:2} establishes that gradients are bounded, enabling their constraint within a constant range. Consequently, we derive an iterative relationship: the sensitivity at epoch $E$ relates to the sensitivity at epoch $E-1$ plus a constant. Iterating this relationship back to epoch 0, where both models remain untrained on the dataset, results in identical initial outputs and zero differences. Only the accumulated gradient term, bounded by the constant term, persists. This accumulated term thus serves as the upper bound for the sensitivity within the DP framework for the model, concluding the proof.

\section{Proof of Theorem~\ref{thm:mydp}} 
\label{sec:appendixB}

Due to the page limit, we only provide the proof sketch here. More details can be found in the history-revised version.

\textbf{Proof Sketch of Theorem~\ref{thm:mydp}:}
Building upon the result from the book~\cite{dwork2014algorithmic}, we continue the derivation for the differential privacy noise of the Gaussian mechanism. Utilizing the tail bound of the Gaussian distribution for bounding, we obtain an expression consisting of a logarithmic term and a quadratic term, requiring their sum to exceed $\ln{\frac{2}{\sqrt{2\pi}\delta}}$. We analyze two distinct cases: $\sqrt{\frac{2}{\pi}} \leq \delta < 1$ and $0 < \delta < \sqrt{\frac{2}{\pi}}$. For the first case, the right-hand side of the constraint inequality derived from the tail bound is negative, while the quadratic term on the left-hand side is always positive. Consequently, it suffices to ensure the logarithmic term on the left is greater than zero, leading to a constrained quadratic inequality from which the result for this case readily follows. For the second case, the right-hand side of the constraint inequality is positive, necessitating a tighter bound. Under this condition, the logarithmic term on the left is readily seen to be always positive; thus, we only require the quadratic term to exceed the right-hand side. This again yields a constrained quadratic inequality, from which the result for this case is obtained. By examining these two scenarios separately, we arrive at the final conclusion of the theorem, concluding the proof.

\section{Proof of Theorem~\ref{Thm:conv}}
\label{sec:appendixC}

Due to the page limit, we only provide the proof sketch here. More details can be found in the history-revised version.

\textbf{Proof Sketch of Theorem~\ref{Thm:conv}:}
We first consider the difference between the two noisy local models from adjacent rounds. Expanding one local model via gradient descent allows us to express this difference as the sum of a negative gradient term and the noise difference between the two adjacent rounds. Leveraging this insight, and invoking Assumption~\ref{Ass:1} along with the Cauchy-Schwarz inequality, we can bound the difference in the loss values of the noisy models across adjacent rounds by the sum of a gradient term, a noise term, and a cross-product term. Subsequently, based on Corollary~\ref{coro:1}, Assumption~\ref{Ass:3}, and the AM-GM inequality, we further bound this difference by the squared gradient term and a constant. Reapplying Assumption~\ref{Ass:1} and subtracting the ideal optimum value $F^*$ from both sides of the inequality allows us to relate the difference between the noisy local models and $F^*$ at round $t$ to that at round $t-1$. Iterating this relationship from round $0$ to round $t-1$ demonstrates that the difference between the noisy local model at round $t$ and $F^*$ is bounded by the initial difference at round $0$ plus a constant term. The model at round $0$ is the initialization, which qualifies as a constant; hence, the upper bound on the difference between the noisy local model at round $t$ and $F^*$ is constrained by a constant value. Building upon this result and utilizing the aggregation process of the FedAvg strategy, we bound the difference between the global model at round $t$ and $F^*$ by a constant, thereby establishing the conclusion of the theorem, concluding the proof.

\begin{table*}[!t]
    \caption{Final accuracy performance of \tool and baselines. ``Rate'' means the accuracy improvement rate of \tool compared to each SOTA work.}
    \centering
    \scriptsize
    \setlength{\tabcolsep}{3pt}
    \begin{tabular}{c|c|c|c c c c c c|c c c c c c} 
    \hline
    \multicolumn{3}{c|}{Scenarios} & \multicolumn{6}{c|}{\textbf{ResNet-18@CIFAR-10}}& \multicolumn{6}{c}{\textbf{ResNet-18@CIFAR-100}}\\\hline
    \multicolumn{3}{c|}{Method} & \makecell{\tool \\ (Ours)} & Full DP & \makecell{Time-\\Varying DP} & \makecell{Sensitive \\ DP} & DPA LDP & \makecell{AdapLDP} & \makecell{\tool \\ (Ours)} & Full DP & \makecell{Time-Varying \\DP} & \makecell{Sensitive \\ DP} & DPA LDP & \makecell{AdapLDP} \\\hline
    \multirow{8}{*}{$\epsilon$} & \multirow{2}{*}{0.2} & Acc. (\%) & \textbf{77.82} & 58.35& 63.65& 64.43 & 66.54 & 73.90 & \textbf{41.75}& 10.13& 22.32& 29.28 & 36.18& 38.47\\
    & & Rate (\%)& - & 33.37 & 22.26 & 20.78 & 16.95& 5.30 & -& 312.14 & 87.05 & 42.59 & 15.40 & 8.53\\\cline{2-15}
    & \multirow{2}{*}{0.3} & Acc. (\%) & \textbf{82.29} & 58.58& 62.54& 67.64 & 73.77& 77.12 & \textbf{50.93}& 22.14& 33.61&  40.27 & 43.52& 48.89\\
    & & Rate (\%)& - & 40.47 & 31.58 & 21.66 & 11.55& 6.70 & -& 130.04 & 51.53 & 26.47 & 17.03 & 4.17 \\\cline{2-15}
    & \multirow{2}{*}{0.4} & Acc. (\%) & \textbf{87.12} & 60.71& 66.79& 69.68 & 77.94 & 83.28 & \textbf{53.42}& 24.18& 39.13& 42.39 & 46.75& 50.86\\
    & & Rate (\%)&- & 43.50 & 30.44 & 25.03 & 11.78 & 4.61 & - & 120.93 & 36.52 & 26.02 & 14.27 & 5.03 \\\cline{2-15}
    & \multirow{2}{*}{0.5} & Acc. (\%) & \textbf{90.06} & 70.59& 78.51& 84.68 & 79.36 & 88.39 & \textbf{61.77}& 41.2& 49.35&  55.53 & 53.04& 60.43\\
    & & Rate (\%)& - & 27.58 & 14.71 & 6.35 & 13.48 & 1.89 & -& 49.93 & 25.17 & 11.24 & 16.46 & 2.22\\\hline
    \multicolumn{3}{c|}{\makecell{Average \\ Improvement Rate}} & - & 36.23\% & 24.75\% & 18.46\% & 13.44\%& 4.63\% & -& 153.26\% & 50.07\% & 26.58\% & 15.79\%& 4.99\% \\\hline
    \multicolumn{3}{c|}{Scenarios} & \multicolumn{6}{c|}{\textbf{CNN@CIFAR-10}}& \multicolumn{6}{c}{\textbf{CNN@CIFAR-100}}\\\hline
    \multicolumn{3}{c|}{Method} & \makecell{\tool \\ (Ours)} & Full DP & \makecell{Time-Varying \\DP} & \makecell{Sensitive \\ DP} & DPA LDP & \makecell{AdapLDP} & \makecell{\tool \\ (Ours)} & Full DP & \makecell{Time-Varying \\DP} & \makecell{Sensitive \\ DP} & DPA LDP & \makecell{AdapLDP} \\\hline
    \multirow{8}{*}{$\epsilon$} & \multirow{2}{*}{0.2} & Acc. (\%) & \textbf{38.81} & 10.00& 10.00 & 12.23 & 22.97& 35.79 & \textbf{15.69}& 1.23 & 1.87 & 2.87 & 7.73 & 12.04\\
    & & Rate (\%)& - & 288.10 & 288.10 & 217.33 & 68.96 & 8.4A & -& 1175.61 & 739.04 & 446.69 & 102.98 & 30.32 \\\cline{2-15}
    & \multirow{2}{*}{0.3} & Acc. (\%) & \textbf{52.13} & 38.15& 42.31& 13.45 & 46.82 & 51.60 & \textbf{18.56}& 1.52& 2.03& 3.55 & 12.56 & 16.88 \\
    & & Rate (\%)& - & 36.64 & 23.21 & 287.58 & 11.34 & 1.03 & -& 1121.05 & 814.29 & 422.82 & 47.77 & 9.95\\\cline{2-15}
    & \multirow{2}{*}{0.4} & Acc. (\%) & \textbf{68.93} & 50.75& 52.37& 64.11 & 60.37 & 67.43 & \textbf{39.45}& 10.83& 22.44& 25.12 & 33.14 & 37.29 \\
    & & Rate (\%)& - & 35.82 & 31.62 & 7.52 & 14.18 & 2.22 & -& 264.27 & 75.80 & 57.05 & 19.04 & 5.79 \\\cline{2-15}
    & \multirow{2}{*}{0.5} & Acc. (\%) & 75.36 & 63.19& 69.74& 71.07 & 70.66 & \textbf{76.49} & \textbf{40.38}& 22.13& 28.93& 35.57 & 35.11& 39.26\\
    & & Rate (\%)& - & 19.26 & 8.06 & 6.04 & 6.65 & -1.48 & -& 82.47 & 39.58 & 13.52 & 15.01 & 2.85\\\hline
    \multicolumn{3}{c|}{\makecell{Average \\ Improvement Rate}} & - & 94.96\% & 87.75\% & 129.62\% & 25.28\%& 2.55\% & -& 660.85\% & 417.18\% & 235.02\% & 46.20\%& 12.23\% \\\hline
    \end{tabular}
    \label{tab:new_appendix_acc_results}
    \vspace{-2ex}
\end{table*}

\begin{table*}[!t]
    \caption{Noise scale performance of \tool and baselines. ``Rate'' means the noise reduction rate of \tool compared to each SOTA work.}
    \centering
    \scriptsize
    \setlength{\tabcolsep}{3pt}
    \begin{tabular}{c|c|c|c c c c c c|c c c c c c}
    \hline
    \multicolumn{3}{c|}{Scenarios} & \multicolumn{6}{c|}{\textbf{ResNet-18@CIFAR-10}}& \multicolumn{6}{c}{\textbf{ResNet-18@CIFAR-100}}\\\hline
    \multicolumn{3}{c|}{Method} & \makecell{\tool \\ (Ours)} & Full DP & \makecell{Time-\\Varying DP} & \makecell{Sensitive \\ DP} & DPA LDP & \makecell{AdapLDP} & \makecell{\tool \\ (Ours)} & Full DP & \makecell{Time-Varying \\DP} & \makecell{Sensitive \\ DP} & DPA LDP & \makecell{AdapLDP} \\\hline
    \multirow{8}{*}{$\epsilon$} & \multirow{2}{*}{0.2} & Noise Scale & 275,447 & 727,901& 862,998 & \textbf{242,583}& 386,453 & 735,138 & 302,859 & 787,542 & 532,786& \textbf{281,766} & 379,128 & 762,930\\
    & & Rate (\%)& - & 62.16 & 68.08 & -13.55 & 28.72 & 62.53 & -& 61.54 & 43.16 & -7.49 & 20.12 & 60.30 \\\cline{2-15}
    & \multirow{2}{*}{0.3} & Noise Scale & \textbf{132,365}& 505,266& 295,231& 165,234 & 261,698 & 512,366 & 167,884 & 543,256 & 324,836& \textbf{127,458} & 274,038 & 538,719\\
    & & Rate (\%)& - & 73.80 & 55.17 & 19.89 & 49.42 & 74.17 & -& 69.10 & 48.32 & -31.72 & 38.74 & 68.84\\\cline{2-15}
    & \multirow{2}{*}{0.4} & Noise Scale & \textbf{87,563} & 383,945& 224,129& 95,468 & 241,778 & 402,970 & 101,256 & 436,562 & 192,879& \textbf{97,145} & 239,767 & 399,175\\
    & & Rate (\%)& - & 77.19 & 60.93 & 8.28 & 63.78 & 78.27 & -& 76.81 & 47.50 & -4.23 & 57.77 & 74.63 \\\cline{2-15}
    & \multirow{2}{*}{0.5} & Noise Scale & \textbf{51,658} & 203,968& 193,487& 53,997 & 156,947 & 217,344 & 94,253 & 267,569 & 122,365& \textbf{88,766} & 142,834 & 232,067\\
    & & Rate (\%)& - & 74.67 & 73.30 & 4.33 & 67.09 & 76.23 & -& 64.77 & 22.97 & -6.18 & 34.01 & 59.39\\\hline
    \multicolumn{3}{c|}{\makecell{Average \\ Reduction Rate}} & - & 71.96\% & 64.37\% & 4.74\% & 52.25\% & 72.80\% & -& 68.06\% & 40.49\% & -12.40\% & 37.66\% & 65.79\% \\\hline
    \multicolumn{3}{c|}{Scenarios} & \multicolumn{6}{c|}{\textbf{CNN@CIFAR-10}}& \multicolumn{6}{c}{\textbf{CNN@CIFAR-100}}\\\hline
    \multicolumn{3}{c|}{Method} & \makecell{\tool \\ (Ours)} & Full DP & \makecell{Time-\\Varying DP} & \makecell{Sensitive \\ DP} & DPA LDP & \makecell{AdapLDP} & \makecell{\tool \\ (Ours)} & Full DP & \makecell{Time-Varying \\DP} & \makecell{Sensitive \\ DP} & DPA LDP & \makecell{AdapLDP} \\\hline
    \multirow{8}{*}{$\epsilon$} & \multirow{2}{*}{0.2} & Noise Scale & \textbf{21,747} & 61,523& 44,752& 22,754 & 31,052 & 63,489 & \textbf{22,968} & 77,523 & 50,803& 23,854 & 35,603 & 72,781\\
    & & Rate (\%)& - & 64.65 & 51.41 & 4.43 & 29.97 & 65.75 & -& 70.37 & 54.79 & 3.71 & 35.49 & 68.44 \\\cline{2-15}
    & \multirow{2}{*}{0.3} & Noise Scale & \textbf{15,870} & 48,236& 30,983& 17,038 & 24,647 & 45,913 & \textbf{19,297} & 52,652 & 34,114& 21,271 & 37,329 & 56,700\\
    & & Rate (\%)& - & 67.10 & 48.78 & 6.86 & 35.61 & 65.43 & -& 63.35 & 43.43& 9.28 & 48.31 & 65.97\\\cline{2-15}
    & \multirow{2}{*}{0.4} & Noise Scale & \textbf{8,798} & 34,717& 18,156& 9,047 & 14,029 & 36,140 & \textbf{9,028} & 42,236 & 23,578& 9,548 & 27,668 & 44,469\\
    & & Rate (\%)& - & 74.66 & 51.54 & 2.75 & 37.29 & 75.66 & -& 78.62 & 61.71 & 5.45 & 67.37 & 79.70\\\cline{2-15}
    & \multirow{2}{*}{0.5} & Noise Scale & \textbf{7,166} & 19,852& 11,549& 8,378 & 9,033 & 20,149 & \textbf{8,677} & 25,786 & 18,270& 8,933 & 10,119 & 18,574\\
    & & Rate (\%)& - & 63.90 & 37.94 & 14.47 & 20.67 & 64.43 & -& 66.35 & 52.51 & 2.87 & 14.25 & 53.28\\\hline
    \multicolumn{3}{c|}{\makecell{Average \\ Reduction Rate}} & - & 67.58\% & 47.42\% & 7.12\% & 30.88\% & 67.82\% & -& 69.68\% & 53.11\% & 5.33\% & 41.35\% & 66.85\% \\\hline
    \end{tabular}
    \label{tab:new_appendix_noise_results}
\end{table*}

\begin{table*}[!t]
    \centering
    \caption{Statistical significance (p-values) of different metrics between \tool and baselines.}
    \scriptsize
    \begin{threeparttable}
    \begin{tabular}{c|l|c|c|c|c|c}
    \hline
    \textbf{Metrics} & \textbf{Scenario} & \textbf{Full DP} & \textbf{\makecell{Time-Varying \\ DP}} & \textbf{Sensitive DP} & \textbf{DPA LDP} & \textbf{AdapLDP} \\\hline
    \multirow{4}{*}{Accuracy} & ResNet-18@CIFAR-10 & 1.54e-6 & 2.64e-5 & 4.61e-5 & 3.02e-4 & 4.23e-2\\
    & ResNet-18@CIFAR-100 & 1.86e-6 & 2.64e-5 & 1.67e-2 & 2.26e-2 & 1.27e-1\\
    & CNN@CIFAR-10 & 2.73e-3 & 1.93e-2 & 3.74e-2 & 2.35e-2 & 5.34e-2\\
    & CNN@CIFAR-100 & 3.21e-4 & 9.35e-3 & 7.34e-3 & 5.91e-3 & 7.97e-2\\\hline
    & Average & 7.64e-4 & 7.18e-3 & 1.54e-2 & 1.31e-2 & 7.56e-2\\\hline
    \multirow{4}{*}{Noise scale} & ResNet-18@CIFAR-10 & 2.20e-6 & 3.19e-4 & 5.59e-1 & 2.73e-3 & 2.20e-6\\
    & ResNet-18@CIFAR-100 & 2.20e-6 & 2.27e-3 & 2.35e-1 & 1.67e-2 & 2.20e-6\\
    & CNN@CIFAR-10 & 2.64e-5 & 2.49e-3 & 2.83e-1 & 2.04e-2 & 2.24e-5\\
    & CNN@CIFAR-100 & 1.54e-6 & 5.64e-4 & 3.37e-1 & 3.19e-4 & 2.38e-4\\\hline
    & Average & 8.09e-6 & 1.41e-3 & 3.54e-1 & 1.01e-2 & 6,62e-5\\\hline
    \multirow{4}{*}{\makecell{Accumulative \\ privacy budget}} & ResNet-18@CIFAR-10 & 2.64e-5 & 2.24e-6 & 3.49e-3 & 4.67e-6 & 2.64e-5 \\
    & ResNet-18@CIFAR-100 & 2.64e-5 & 1.54e-6 & 1.85e-2 & 1.86e-6 & 2.64e-5\\
    & CNN@CIFAR-10 & 2.64e-5 & 5.58e-6& 3.64e-2 & 3.24e-6 & 2.64e-5\\
    & CNN@CIFAR-100 & 2.64e-5  & 1.55e-6& 3.32e-2 & 6.68e-6 & 2.64e-5\\\hline
    & Average & 2.64e-5  & 2.73e-6 & 2.29e-2 & 4.11e-6 & 2.64e-5\\\hline
    \end{tabular}
    \label{tab:pvalue}
    \begin{tablenotes}
        \footnotesize
        \item[*] All p-values were computed based on experimental data from the final five rounds. For each algorithm and scenario, the final five rounds of data were selected to enhance statistical power and reliability, thus mitigating the risk of failing to detect genuine differences due to very small sample sizes.
      \end{tablenotes}
    \end{threeparttable}
\end{table*}

\color{black}
\section{Mathematical Support of the sequential accountant method in Section~\ref{sec:accumulate_privacy_budget}}
\label{app:theoretical_background}

The methods employed in Section~\ref{sec:accumulate_privacy_budget} are grounded in a fundamental theorem from the Composition Theorems of Differential Privacy theory, as illustrated in Theorem~\ref{thm:naive_composition}. The detailed proof of this theorem can be found in reference~\cite{dwork2014algorithmic}.
\begin{theorem}
Let $\mathcal{M}_i:\mathbb{N}^{|\mathcal{X}|}\to\mathcal{R}_i$ be an $(\epsilon_i, 0)$-DP algorithm for $i \in [k]$. Then if $\mathcal{M}_{[k]}:\mathbb{N}^{|\mathcal{X}|}\to\Pi_{i=1}^k\mathcal{R}_i$ is defined to be $\mathcal{M}_{[k]}(x)=(\mathcal{M}_1(x), ..., \mathcal{M}_k(x))$, then $\mathcal{M}_{[k]}$ is $(\sum_{i=1}^k\epsilon_i, 0)$-Differentially Private.
    \label{thm:naive_composition}
\end{theorem}

This theorem can be further extended to $(\epsilon, \delta)$-DP, as detailed in Corollary~\ref{coro:naive_composition}.
\begin{corollary}[Navie Composition Theorem]
Let $\mathcal{M}_i:\mathbb{N}^{|\mathcal{X}|}\to\mathcal{R}_i$ be an $(\epsilon_i, \delta_i)$-DP algorithm for $i \in [k]$. Then if $\mathcal{M}_{[k]}:\mathbb{N}^{|\mathcal{X}|}\to\Pi_{i=1}^k\mathcal{R}_i$ is defined to be $\mathcal{M}_{[k]}(x)=(\mathcal{M}_1(x), ..., \mathcal{M}_k(x))$, then $\mathcal{M}_{[k]}$ is $(\sum_{i=1}^k\epsilon_i, \sum_{i=1}^k\delta_i)$-Differentially Private.
    \label{coro:naive_composition}
\end{corollary}

The sequential accountant method is a statistical approach based on Corollary~\ref{coro:naive_composition}. In particular, if the Gaussian noise injected during the training process satisfies the condition in Theorem~\ref{thm:2}, each noise injection step guarantees $(q\epsilon, q\delta)$-DP, where $q$ is the sampling probability based on the batch size. After $T$ iterations, various advanced composition theorems can be applied to derive different privacy bounds, such as the $(qT\epsilon, qT\delta)$-DP. This framework allows a rigorous quantification of privacy guarantees in iterative training processes under differential privacy. 

The detailed quantitative steps of the sequential accountant method used in Section~\ref{sec:accumulate_privacy_budget} are as follows:
\begin{enumerate}
    \item Determine the specific privacy budget consumed by the algorithm during the protection process by deriving the actual variance of the injected noise distribution.
    \begin{equation}
        \sigma = \sqrt{2\log\frac{1.25}{\delta}}/\epsilon.
    \end{equation}
    \item Determine the $q$ by the batch size of the training process.
    \begin{equation}
        q = \frac{Batch\_Size}{|\mathcal{D}_i|}.
    \end{equation}
    \item Determine the iteration number $T$.
    \item Determine the cumulative privacy budget by the sequential accountant method.
    \begin{equation}
        O(\epsilon) = qT\epsilon.
    \end{equation}
\end{enumerate}

\section{More Experimental Results}

\label{app:more_results}

We present all the numerical results of \tool and SOTA works (including Full DP, Time-Varying DP, and Sensitive DP) across all experimental scenarios. Table~\ref {tab:new_appendix_acc_results} shows the final accuracies of each method under different scenarios as well as the accuracy improvement rates of \tool compared to SOTA works. Table~\ref {tab:new_appendix_noise_results} shows the total noise scales of each method under different scenarios as well as the reduction ratios of \tool compared to SOTA works.

Based on Table~\ref{tab:new_appendix_acc_results}, we can calculate the average improvement rate of \tool compared to each SOTA work. Using the average improvement rate compared to Full DP as an example, we averaged all the average improvement rates across all experimental scenarios as follows:

\begin{small}
\begin{equation}
    (36.23\% + 153.26\% + 94.96\% + 660.85\%)/4 = 236.32\%, 
\end{equation}
\end{small}
which is consistent with the value in the Introduction. 

Similarly, we have the average improvement rate compared to Time-Varying DP, Sensitive DP, DPA LDP, and AdapLDP as 144.93\%, 102.42\%, 25.18\%, and 6.10\%. Then, we averaged the improvement rates of these three algorithms to obtain the overall average improvement value:

\begin{small}
    \begin{equation}
     (236.32\% + 144.93\% + 102.42\% + 25.18\% + 6.10\%)/5 = 102.99\%.
\end{equation}
\end{small}

We can calculate the average noise reduction rate with the same method based on Table~\ref{tab:new_appendix_noise_results}. Using the average reduction rate compared to Full DP as an example, we averaged all the average reduction rates across all experimental scenarios as follows:

\begin{small}
\begin{equation}
        (71.96\% + 68.06\% + 67.58\% + 69.68\%)/4 = 69.32\%,
\end{equation}
\end{small}
which is consistent with the value in the Introduction. 

Similarly, we have the average reduction rate compared to Time-Varying DP, Sensitive DP, DPA LDP, and AdapLDP as 51.35\%, 1.20\%, 40.54\%, and 68.31\%. Then, we averaged the reduction rates of these three algorithms to obtain the overall average reduction value:

\begin{small}
\begin{equation}
     (69.32\% + 51.35\% + 1.20\% + 40.54\% + 68.31\%)/5 = 46.14\%.
\end{equation}
\end{small}

\begin{table}[!t]
    \centering
    \caption{Accuracy and training durations performance of \tool and baselines when training ResNet-152 on CIFAR-100.}
    \scriptsize
    \setlength{\tabcolsep}{1pt}
    \begin{tabular}{c|c|c|c|c|c|c|c}
    \hline
    \multicolumn{2}{c|}{Method} & \makecell{\tool \\ (Ours)} & Full DP & \makecell{Time-\\Varying DP} & \makecell{Sensitive \\ DP} & DPA LDP & \makecell{AdapLDP} \\ \hline
    \multirow{2}{*}{$\epsilon=0.2$} & Acc. (\%) & 82.44 & 45.18 & 62.99 & 68.43 & 72.05 & 75.31 \\
    & Duration (s) & 257,104 & 239,816 & 242,993 & 247,302 & 254,731 & 263,998 \\\hline
    \multirow{2}{*}{$\epsilon=0.5$} & Acc. (\%) & 89.40 & 72.33 & 77.59 & 74.81 & 83.57 & 86.12 \\
    & Duration (s) & 256,943 & 239,407 & 243,434 & 247,759 & 254,003 & 265,376\\\hline
    \end{tabular}
    \label{tab:res152}
\end{table}

\begin{table}[!t]
    \centering
    \caption{The comparative quality of privacy data generation by adversaries when confronting \tool versus baseline privacy-preserving algorithms across diverse privacy data types.}
    \scriptsize
    \setlength{\tabcolsep}{3.5pt}
    \begin{threeparttable}
    \begin{tabular}{c|l|c c c c c c}
    \hline
     \multirow{2}{*}{\makecell{Privacy \\ Budget}} &\multirow{2}{*}{Method} & \multicolumn{6}{c}{\textbf{FID score across diverse labels}}  \\\cline{3-8}
     & & \makecell{car} & \makecell{horse} & \makecell{\textcolor{red}{bird}} & \makecell{A} & \makecell{e} & \makecell{\textcolor{red}{5}}\\\hline
     - & No protection & 33.98 & 40.16 & 52.77 & 21.93 & 17.40 & 24.35 \\\hline
     \multirow{6}{*}{0.2}& Full DP & 79.14 & 77.38 & 144.62 & 47.95 & 44.14 & 83.06\\

     & Time-Varying DP & 67.89 & 69.13 & 108.37 & 39.22 & 30.15 & 59.38\\

     & Sensitive DP & 70.44 & 78.86 & 139.02 & 47.37 & 43.99 & 80.57\\

     & DPA LDP & 58.44 & 59.30 & 119.79 & 36.35 & 31.44 & 72.80\\

     & AdapLDP & 42.37 & 48.29 & 79.36 & 33.63 & 34.52 & 46.39\\

     & \tool & 68.33 & 72.79 & 147.22 & 34.80 & 28.66 & 93.77 \\\hline
     \multirow{6}{*}{0.5}& Full DP & 49.33 & 42.07 & 76.68 & 26.79 & 28.40 & 49.14\\

     & Time-Varying DP & 51.29 & 43.33 & 67.92 & 24.57 & 25.88 & 37.42\\

     & Sensitive DP & 53.77 & 50.98 & 81.49 & 34.05 & 33.91 & 62.86\\

     & DPA LDP & 46.25 & 45.58 & 62.46 & 22.73 & 21.99 & 44.50\\

     & AdapLDP & 39.42 & 44.41 & 55.72 & 20.80 & 24.39 & 29.17\\

     & \tool & 46.79 & 48.03 & 94.82 & 27.74 & 37.46 & 72.48\\\hline
    \end{tabular}
    \begin{tablenotes}
        \footnotesize
        \item[*] The FID scores are computed based on 10,000 adversarially reconstructed privacy images per label, with all tabulated experimental results representing averages derived from three independent runs.
        \item[*] Higher FID scores indicate greater semantic incoherence in generated images, signifying enhanced privacy protection performance that effectively prevents adversaries from successfully reconstructing private data.
      \end{tablenotes}
    \end{threeparttable}
    \label{tab_attack_gen_quality}
\end{table}

We present the results of the statistical significance (p-values) of prediction accuracy comparisons, noise scale comparisons, and cumulative privacy budget comparisons between \tool and baselines in Table~\ref{tab:pvalue}.

The experimental results of the ResNet-152 are presented in Table~\ref{tab:res152}, demonstrating the scalability of \tool on the deep model.

Table~\ref{tab_attack_gen_quality} presents the FID scores when confronting the \tool's defense and baselines' defense mechanisms, illustrating the quality of the reconstructed privacy data, thereby showcasing the protection capabilities of \tool and baselines.